\documentclass[a4paper,UKenglish,cleveref,nameinlink,autoref,thm-restate]{lipics-v2021}

\usepackage{apptools}
\newcommand{\restateref}[1]{\IfAppendix{\hyperref[#1]{$\star$}}{\hyperref[#1*]{$\star$}}}

\Crefname{theorem}{Thm.}{Thms.}
\Crefname{figure}{Fig.}{Figs.}

\usepackage{comment}
\hideLIPIcs  %uncomment to remove references to LIPIcs series (logo, DOI, ...), e.g. when preparing a pre-final version to be uploaded to arXiv or another public repository

\usepackage{paralist}

\usepackage{array}
\usepackage{wrapfig}

\title{Planar Stories of Graph Drawings: Algorithms and Experiments} 

\author{Carla Binucci}{University of Perugia, Italy }{carla.binucci@unipg.it}{https://orcid.org/0000-0002-5320-9110}{}%TODO mandatory, please use full name; only 1 author per \author macro; first two parameters are mandatory, other parameters can be empty. Please provide at least the name of the affiliation and the country. The full address is optional. Use additional curly braces to indicate the correct name splitting when the last name consists of multiple name parts.

\author{Sabine Cornelsen}{University of Konstanz, Germany}{sabine.cornelsen@uni-konstanz.de}{https://orcid.org/0000-0002-1688-394X}{}

\author{Walter Didimo}{University of Perugia, Italy }{walter.didimo@unipg.it}{https://orcid.org/0000-0002-4379-6059}{}

\author{Seok-Hee Hong}{University of Sydney, Australia }{seokhee.hong@sydney.edu.au}{https://orcid.org/0000-0003-1698-3868}{}

\author{Eleni Katsanou}{National Technical University of Athens, Greece}{ekatsanou@mail.ntua.gr}{https://orcid.org/0000-0002-1001-1411}{}

\author{Maurizio Patrignani}{Roma Tre University, Italy}{maurizio.patrignani@uniroma3.it}{https://orcid.org/0000-0001-9806-7411}{}

\author{Antonios Symvonis}{National Technical University of Athens, Greece}{symvonis@mail.ntua.gr}{
https://orcid.org/0000-0002-0280-741X}{}

\author{Samuel Wolf}{University of W\"urzburg, Germany }{samuel.wolf@uni-wuerzburg.de}{https://orcid.org/0009-0009-7098-6147}{}

\authorrunning{Binucci, Cornelsen, Didimo, Hong, Katsanou, Patrignani, Symvonis, Wolf}
\Copyright{Carla Binucci, Sabine Cornelsen, Walter Didimo, Seok-Hee Hong, Eleni Katsanou, Maurizio Patrignani, Antonios Symvonis, Samuel Wolf} %TODO mandatory, please use full first names. LIPIcs license is "CC-BY";  http://creativecommons.org/licenses/by/3.0/

 %TODO mandatory; please add comma-separated list of keywords

\relatedversion{This is the extended version of C. Binucci, S. Cornelsen, W. Didimo, S.-H. Hong, E. Katsanou, M. Patrignani, A. Symvonis, S. Wolf, ``Planar Stories of Graph Drawings: Algorithms and Experiments'', to appear in the Proc. of the 33rd International Symposium on Graph Drawing and Network Visualization, GD 2025, LIPIcs, Volume 357, 2025.} %optional, e.g. full version hosted on arXiv, HAL, or other respository/website

\funding{Carla Binucci and Walter Didimo were partially supported by the Italian Ministry of University and Research (MUR) under the PRIN Project no. 2022TS4Y3N – EXPAND and the PON Project RASTA, ARS01\_00540; Maurizio Patrignani was partially supported by the Italian Ministry of University and Research (MUR) under PRIN Project no. 2022TS4Y3N - EXPAND.}%optional, to capture a funding statement, which applies to all authors. Please enter author specific funding statements as fifth argument of the \author macro.

\acknowledgements{This work started at the Bertinoro Workshop on Graph Drawing BWGD 2025.}

\nolinenumbers %uncomment to disable line numbering

\usepackage[utf8]{inputenc}
\usepackage{xspace}
\usepackage{todonotes}
\usepackage{thmtools} 
\usepackage{apptools}
\usepackage{complexity}
\usepackage[inline]{enumitem}
\usepackage{wrapfig}

\crefname{property}{Property}{Properties}
\usepackage{graphicx} % Required for inserting images
\begin{document}

\maketitle

\begin{abstract}
We address the problem of computing a dynamic visualization of a geometric graph~$G$ as a sequence of frames. Each frame shows only a portion of the graph but their union covers $G$ entirely. The two main requirements of our dynamic visualization are: $(i)$ guaranteeing drawing stability, so to preserve the user's mental map; $(ii)$ keeping the visual complexity of each frame low. 
To satisfy the first requirement, we never change the position of the vertices. Regarding the second requirement, we avoid edge crossings in each frame. More precisely, in the first frame we visualize a suitable subset of non-crossing edges; in each subsequent frame, exactly one new edge enters the visualization and all the edges that cross with it are deleted. We call such a sequence of frames a \emph{planar story} of~$G$.
Our goal is to find a planar story whose minimum number of edges contemporarily displayed is maximized (i.e., a planar story that maximizes the minimum frame size). Besides studying our model from a theoretical point of view, we also design and experimentally compare different algorithms,
both exact techniques and heuristics. These algorithms provide an array of alternative trade-offs between efficiency and effectiveness, also depending on the structure of the input graph.   
\end{abstract}

\section{Introduction}\label{se:intro}
Conveying the structure of a graph in a single visualization is the goal of many graph drawing algorithms. Among these, force-directed algorithms are the most popular and widely used in practice~\cite{DBLP:reference/crc/Kobourov13}. They produce layouts in which vertices are represented as points in the plane and edges as straight-line segments. However, when a graph is complex or locally dense, the layout may feature numerous edge crossings, which significantly reduce its readability, as also witnessed by several human cognitive experiments in graph drawing~\cite{DBLP:conf/gd/Purchase97,DBLP:journals/iwc/Purchase00,DBLP:journals/ese/PurchaseCA02,DBLP:journals/tvcg/PurchasePP12}.

This scenario naturally motivates strategies that either attempt to simplify the (single) graph visualization, for instance through graph sampling~\cite{Eades,BCproxy,proxy,r-05,DBLP:journals/tvcg/WuCASQC17,zxyq-13} or edge bundling~\cite{DBLP:conf/gd/ArchambaultLNPT24,DBLP:journals/cgf/HoltenW09,Nguyen}, or that attempt to ``distribute'' the nodes and edges of the graph across multiple visualizations, i.e., a sequence of frames, each showing only a portion of the graph (see, e.g.,~\cite{DBLP:journals/jcss/BinucciGLLMNS24,DBLP:journals/jgaa/BorrazzoLBFP20,DBLP:journals/jgaa/BattistaDGGOPT23,DBLP:journals/vlc/GiacomoDLMT14}). While the first type of strategy focuses on a static visualization that may cause ambiguity or loss of information, computing a sequence of visualization frames allows for a dynamic exploration of the entire graph, without altering its structure. However, it  should guarantee \emph{drawing stability} (i.e., the layout of the portion of the graph shared by two consecutive frames should remain unchanged) and each frame should have \emph{low visual complexity}, for example by guaranteeing planarity or few edges crossings.

\subparagraph{Contribution.} In this paper we address the second type of strategy mentioned above. We introduce a new model for the dynamic visualization of a static graph as a sequence of frames. This model assumes that the input is a geometric graph $G$ (i.e., a graph drawn in the plane with straight-line edges, for example by some force-directed algorithm) and aims to generate a sequence of visualization frames such that: $(i)$ the vertices of~$G$ are all present in any frame, always at their initial input coordinates; $(ii)$ the initial frame contains a  subset of edges that yields a planar subgraph; $(iii)$ in each subsequent frame, exactly one new edge~$e$ enters the visualization and the minimum subset of edges of the current frame that together with~$e$ violate planarity disappear; $(iv)$ when an edge leaves the visualization, it no longer appears in the future; $(v)$ the union of all frames contains the whole edge set of~$G$. We call such a sequence of frames a \emph{planar story} of~$G$ (see \cref{se:preliminaries} for a formal definition).

A natural objective function when computing a planar story is that the amount of edges in any frame is not too small. More precisely, we are interested in planar stories whose minimum number of edges contemporarily displayed (across the entire sequence of frames) is maximized. We call such a story \emph{min-frame optimal}. 
Besides studying our model from a theoretical point of view, we design and experimentally compare different algorithms (both exact techniques and heuristics). Our main results are as follows:

\begin{itemize}
\item We show in \cref{se:heuristic} that a min-frame optimal planar story of a geometric graph $G$ can be efficiently computed when $G$ is 2-plane (i.e., each edge is crossed at most twice), although   
%we show that this result may not be easily extendable to 3-planar graphs\todo{SC: I now removed this theorem from the appendix.} and we prove that 
the problem is NP-hard in the general case (\cref{se:complexity}). 

\item We describe several heuristics based on a common greedy framework; all of them work for general graphs (\cref{se:heuristic}). They differ in the strategy for computing the initial and the last frames, and in the criteria used to select the edge that enters the visualization in each subsequent frame. %For general graphs, w
We also describe an Integer Linear Program (ILP) that solves the problem optimally (\cref{se:ilp}).

\item We discuss the results of an extensive experimental analysis that compares our heuristics and the exact algorithm. The results highlight algorithmic trade-offs between efficiency and effectiveness, and show how some of our heuristics are able to achieve the optimum in many cases in a short runtime. See \cref{se:experiments}.   
\end{itemize}

The paper concludes by discussing some future research directions (\cref{se:conclusions}). 
Due to space limitations some proofs and technical details are moved to the long version of the paper appendix.
%Due to space limitations some proofs and technical details can be found in~\cite{arxiv}.

\section{Related Work}\label{related-work}

The research in this paper is inspired by a recent topic in graph drawing named \emph{graph stories}. A graph story corresponds to a temporal sequence of frames, each displaying only part of the graph and whose union covers all nodes and edges of the graph. The term ``graph story'' was introduced by Borrazzo, Da Lozzo, Di Battista,  Frati, and Patrignani~\cite{DBLP:journals/jgaa/BorrazzoLBFP20}. 
Different from our scenario, in their model no drawing of the graph is provided in advance. The vertices enter the visualization one at a time and persist in the visualization for a fixed amount of time (i.e., for a given subsequence of frames). 
When a vertex enters the visualization its coordinates have to be computed.
In any frame, the user sees only the graph induced by the displayed vertices; the drawing in each frame must be straight-line and planar, and the layout of the graph portion shared by two consecutive frames cannot change. The authors in~\cite{DBLP:journals/jgaa/BorrazzoLBFP20} provide bounds on the area requirement of graph stories for paths and trees, and exhibit planar graphs that do not admit a graph story within their rules. The model in~\cite{DBLP:journals/jgaa/BorrazzoLBFP20}  has been further studied in~\cite{DBLP:journals/jgaa/BattistaDGGOPT23,DBLP:conf/gd/BattistaDGGOPT22}, by allowing edges to be Jordan curves rather than straight-line segments. Also, several authors proposed a different setting in which each vertex persists in the visualization until all its neighbors have been displayed~\cite{DBLP:conf/gd/BinucciGLLMNS22,DBLP:journals/jcss/BinucciGLLMNS24,DBLP:conf/sofsem/FialaFLWZ24}. 
In all the aforementioned papers, the main focus is on exploring the time complexity and theoretical questions, often restricting to very specific graph families. Conversely, our model works for general graphs and, besides providing a solid theoretical basis, our goal is to derive practical implementations with different trade-offs between effectiveness and efficiency.    

Another problem related to our research has been studied in~\cite{DBLP:conf/gd/GiacomoDLMT13,DBLP:journals/vlc/GiacomoDLMT14}. Similarly to our setting, it assumes that the input is a straight-line drawing of a graph; the objective is to partition the edge set into the minimum number of frames, such that the drawing in each frame has some desired property (e.g., planarity). However, unlike our model, the visualization frames do not form a temporal sequence and each edge appears in just one~frame.  

Finally, we remark that representing a graph across a sequence of frames (using small multiples or animation approaches) is strongly related to the  problem of visually conveying dynamic graphs (see~\cite{DBLP:journals/cgf/BeckBDW17} for a survey on this topic). However, a dynamic graph evolves over time, and the ordered sequence of elements that appear or disappear during the graph evolution is part of the input for the drawing algorithm. In our setting, the graph is static and the algorithm can decide the temporal sequence of edges for dynamically visualizing it.

\section{Basic Definitions and Preliminary Considerations}\label{se:preliminaries}

Given a graph $G$, we denote by $V(G)$ and $E(G)$ the set of vertices and the set of edges of~$G$, respectively. For a set $V' \subseteq V(G)$, let $N(V')$ be the set of neighbors of vertices in~$V'$ and let~$G-V'$ be the subgraph of~$G$ induced by~$V(G) \setminus V'$.
%
% For a set~$V$ of vertices of~$G$, let~$N(V)$ be the set of neighbors of vertices in~$V$, %let $N[V]=V \cup N(V)$, 
% and  let~$G-V$ be the subgraph of~$G$ induced by~$V(G) \setminus V$.  
%
Let $G$ be a \emph{geometric graph}, that is, each vertex $v \in V(G)$ is a point in the plane and each edge $e=uv \in E(G)$
is the straight-line segment connecting $u$ and $v$. 
A \emph{planar frame} $F$ of $G$ is any subgraph of $G$ consisting of all the vertices of $G$ and a subset of edges of $G$ that do not cross. A \emph{planar story} of $G$ is a sequence $\sigma = \langle F_1, F_2, \dots, F_\tau\rangle$ such that: 

\begin{enumerate}
\item Each $F_i$ is a planar frame of $G$, with $i \in \{1, \dots, \tau\}$.
\item Each edge $e \in E(G)$ appears in a non-empty  subsequence of consecutive planar frames.
%\item The frame $F_1$ contains a subset of non-crossing edges (which corresponds to an independent set of $X$.
\item For each $i \in \{2, \dots, \tau\}$, $F_i$ contains exactly one edge $e$ of $G$ that is not in the union of all frames $F_j$, with %ASymv $j \leq i$,
$j < i$, 
and exactly those edges of $F_{i-1}$ that do not cross $e$.
\end{enumerate}

Let $\Sigma(G)$ be the set of all planar stories of $G$. For a planar story $\sigma = \langle F_1, F_2, \dots, F_\tau\rangle \in \Sigma(G)$, we call $F_1$ and $F_\tau$ the \emph{initial frame} and the \emph{last frame} of $\sigma$, respectively, let $m_i$ be the \emph{size} of $F_i$, that is, the number of edges in~$F_i$. Let $\mu(\sigma) = \min \{m_1, \dots, m_\tau\}$, and let $\mu(G)$ be the maximum $\mu(\sigma)$ over all planar stories $\sigma$ of $G$, i.e., $\mu(G) = \max \{\mu(\sigma) \; | \; \sigma \in \Sigma(G)\}$. If $\mu(\sigma)=\mu(G)$ then $\sigma$ is said to be \emph{min-frame optimal}. 
%We call $\mu(\sigma)$ the \emph{minimum frame size} of $\sigma$ and $\mu(G)$ the \emph{minimum planar size} of $G$.
%
We call \textsc{MaxMinFramePlanarStory($G$)} the problem of computing a min-frame optimal planar story of~$G$.

\begin{remark}\label{re:crossing-edges}
Since all the crossing-free edges of $G$ appear in every frame of any planar story of $G$, we can safely disregard them when solving \textsc{MaxMinFramePlanarStory($G$)}. Hence, without loss of generality, we will assume that $G$ does not contain crossing-free edges.
\end{remark}
%

%FIGURE
\Cref{fig:planar-story-a} shows a geometric graph $G$ and a planar story $\sigma$ of $G$, consisting of six frames and where $\mu(\sigma) = 3$. In the figure, the red edges are those of the initial frame, the blue edges are those of the final frame, and the black edges are the remaining edges. 
\begin{figure}
    \centering
    \includegraphics[page=4]{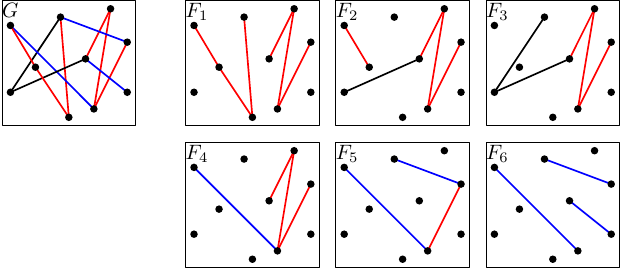}
    \caption{A planar story consisting of 6 frames. The top-leftmost figure shows the input geometric graph $G$; red edges and blue edges are those in the first frame and in the last frame, respectively. The black edges are the remaining edges.}
    \label{fig:planar-story-a}
\end{figure}

\begin{proposition}\label{pr:disjoint}
In any planar story $\sigma = \langle F_1, \dots, F_\tau \rangle$, $F_1$ and $F_\tau$ are edge disjoint. 
\end{proposition}
\begin{proof}
Let $e$ be an edge in $F_1$. By \cref{re:crossing-edges}, $e$ is crossed by some edge $e'$ of the graph. Let $F_i$ be the first frame of $\sigma$ that contains $e'$, with $i \geq 2$.
By definition of planar story, $e$ does not occur in any frame $F_j$ with $j \geq i$. In particular, $e$ is not contained in $F_\tau$.
\end{proof}

Given a planar frame $F$ of $G$, let $\mu_F(G)$ be the maximum $\mu(\sigma)$ over all planar stories $\sigma$ of $G$ with initial frame $F$.
\cref{pr:disjoint} implies the following.

\begin{corollary}\label{cor:bound}
	Let $F$ be a planar frame of $G$, and let $F'$ be a maximum planar subgraph of $G - E(F)$. Then $\mu_F(G) \leq \min\{|E(F)|,|E(F')|\}$. Therefore, $\mu(G) \leq |E(G)|/2$.
\end{corollary}

\subparagraph{Planar Graph Stories and Crossing Graph.}
The \emph{crossing graph} $X$ of a geometric graph $G$ is defined as follows:
%\begin{itemize}
$(i)$ the vertex set $V(X)$ corresponds to the edge set $E(G)$;
$(ii)$ there is an edge $uv \in E(X)$ if and only if the edges corresponding to $u$ and $v$ in $E(G)$~cross~each~other. 
%\end{itemize}
%
The crossing-free edges of $G$ correspond to isolated vertices of $X$; by \cref{re:crossing-edges}, we assume that $X$ has no isolated vertices.
Note that a graph is the %intersection 
crossing graph of a geometric graph if and only if it is a \emph{segment intersection graph} \cite{kratochvilM:jct94}. This includes cliques, but also planar graphs \cite{chalopinG:stoc09,goncalvesIP:soda18}. 
The crossing graph of a geometric graph with $m$ edges and $\chi$ crossings can be computed in $\mathcal O(m\log m + \chi)$~time with the algorithm of Balaban \cite{balaban:socg95} for segment intersection.

\begin{figure}
    \centering
    \includegraphics[page=3]{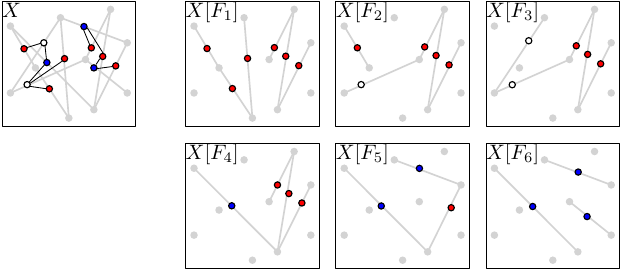}
    \caption{The crossing graph $X$ of the graph $G$ in \cref{fig:planar-story-a} (top-left) and the planar story in \cref{fig:planar-story-a} from the point of view of the crossing
graph $X$ of $G$.}
    \label{fig:planar-story-b}
\end{figure}

If $\sigma=\langle F_1, \dots, F_\tau\rangle$ is a planar story of $G$, each frame $F_i$ ($i \in \{1,\dots,\tau\}$) corresponds to an \emph{independent set} $X[F_i]$ of $X$,
i.e., no two vertices of $X[F_i]$ are adjacent. 
We denote by $\sigma_X = \langle X[F_1], \dots, X[F_\tau]\rangle$ the sequence that describes the planar story $\sigma$ from the point of view of $X$. See \Cref{fig:planar-story-b} for an example.  We will also refer to $\sigma_X$ as a planar story and to the independent sets $X[F_i]$ as (planar) frames. If not stated otherwise, we use red for the initial frame and blue for the final frame in our figures.
A \emph{maximum pair of independent sets} is a pair $(I_1,I_2)$ of two disjoint independent sets such that $\min\{|I_1|,|I_2|\}$ is maximum. 
A maximum pair $(I_1,I_2)$ of independent sets is \emph{Pareto optimal} if also $\max\{|I_1|,|I_2|\}$ is maximum. 
As a further corollary of \cref{pr:disjoint}, we obtain the~following. %bound for $\mu(G)$. 

\begin{corollary}\label{cor:pair}
    Let $G$ be a geometric graph, $X$ be the crossing graph of $G$, and $(I_1,I_2)$ be a \emph{maximum pair of independent sets} of $X$. Then $\mu(G) \leq \min\{|I_1|,|I_2|\}$. 
\end{corollary}

We observe that, in order to get a min-frame optimal story of a geometric graph $G$, it might be necessary to choose as the initial frame a planar subgraph of $G$ that is not a maximal planar subgraph. 
In terms of the crossing graph, this means that it might be necessary to start with an independent set that is not maximal. 
See \cref{fig:notmaximal} for an example  where the crossing graph is a caterpillar and \cref{le:initial-frame-not-maximal} in the appendix for the details.
%non arxiv: See \cref{fig:notmaximal} for an example  where the crossing graph is a caterpillar.

\begin{figure}
	\centering
	\begin{minipage}[b]{0.45\linewidth}
		\centering
		\includegraphics[page=1]{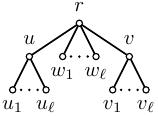}
		\subcaption{crossing graph\label{fig:notmaximal:graph}}
	\end{minipage}\hfil
	\begin{minipage}[b]{0.55\linewidth}
	\centering
	\includegraphics[page=6]{bad_catarpillars}
	\subcaption{non-maximal initial frame, $\mu_F(G) = \mu(G) = \frac32\ell + 1$ \label{fig:notmaximal:better}}
	\end{minipage}\\
	\begin{minipage}{0.25\linewidth}
		\centering
		\includegraphics[page=2]{bad_catarpillars}
		\subcaption{$\mu_F(G) \leq |I'| = \ell+2$ \label{fig:notmaximal:0}}
	\end{minipage}\hfil
		\begin{minipage}{0.25\linewidth}
		\centering
		\includegraphics[page=3]{bad_catarpillars}
		\subcaption{\nolinenumbers$\mu_F(G) \leq |I| = \ell+2$\label{fig:notmaximal:1}}
	\end{minipage}\hfil
		\begin{minipage}{0.25\linewidth}
		\centering
		\includegraphics[page=4]{bad_catarpillars}
		\subcaption{\nolinenumbers$\mu_F(G) \leq |I'| = \ell+1$\label{fig:notmaximal:2}}
	\end{minipage}\hfil	
			\begin{minipage}{0.25\linewidth}
		\centering
		\includegraphics[page=5]{bad_catarpillars}
		\subcaption{\nolinenumbers$\mu_F(G) \leq |I'| = 2$\label{fig:notmaximal:3}}
	\end{minipage}
	
	\caption{\label{fig:notmaximal}
        The crossing graph $X$ of some graph $G$.
		The set $I \subseteq V(X)$ of red vertices corresponds to the initial frame $F$. Blue vertices are a maximum independent set  $I'$ within $X-I$.
		(c)-(f) If we choose $I$ as a maximal independent set, then there is a frame with at most $\min\{|I|,|I'|\} \leq \ell+2$ edges. However, there is a planar story (b) with at least $\frac32\ell + 1$ edges in any frame. Numbers close to the vertices indicate the order in which we add the respective edges into the planar story.} 
\end{figure}

\section{Problem Complexity}\label{se:complexity}

In \cref{se:heuristic} we show how to solve \textsc{MaxMinFramePlanarStory($G$)} optimally in linear time when $G$ is a 2-plane graph and in cubic time if $G$ is a 3-plane graph whose crossing graph has no cycles. We recall that a geometric graph is \emph{$k$-plane} if each edge is crossed at most $k$ times.
However, \textsc{MaxMinFramePlanarStory($G$)} is NP-hard in general. 
To this end, consider its decision version 
\textsc{MaxMinFramePla\-narStoryD($G$,$m$)}: 
Given a graph~$G$ and a target value~$m$, does $G$ admit a planar story $\sigma = \{F_1, F_2, \dots, F_\tau\}$ such that~$\mu(\sigma) \geq m$?

\begin{restatable}
[\restateref{thm:np-complete}]
{theorem}{npcomplete}
    \label{thm:np-complete}
Problem \textsc{MaxMinFramePlanarStoryD($G,m$)} is NP-complete. 
\end{restatable}   
\begin{proof}[Proof Sketch:]
It is easy to see that \textsc{MaxMinFramePla\-narStoryD($G$,$m$)} is in NP.
To prove the hardness,
we apply a reduction from  
\textsc{PlanarMaximumIndependentSetD($H$,$k$)}, i.e., the problem of deciding whether a planar graph $H$ admits an independent set of size at least $k$; this problem is known to be NP-hard~\cite[Theorem 4.1]{m-fcgpa-01}. To this end, let $X$ be the disjoint union of $H$ and a star $S$ with $k+1$ leaves (see \cref{fig:npCompleteness}). Since $X$ is planar, it is a segment intersection graph~\cite{chalopinG:stoc09}, i.e., a crossing graph of some geometric graph $G$. We show that $\mu(G)\geq k+1 =: m$ if and only if $H$ contains an independent set of size at least $k$.

Assume first that $H$ contains an independent set $I$ of size at least $k$. Consider the following planar story. The initial frame contains the edges corresponding to $I$ and the center of the star $S$. The size of the initial frame is $k+1$. For the subsequent frames, we add the edges corresponding to the leaves of $S$, in any arbitrary order. This increases the size of the frame up to $2k+1$. Finally, we add in arbitrary order all edges corresponding to vertices of $H$ that were not in $I$. During this process, the edges corresponding to the leaves of $S$ will never be removed. Thus, the size of each frame is at least $k+1$.

Assume now that there is a planar story $\sigma$ of $G$ with $\mu(\sigma) \geq k+1$. We show that $H$ contains an independent set of size at least $k$. Let $F$ be the frame of $\sigma$ that contains the edge corresponding to the center of $S$. In this frame, there is no edge corresponding to a leaf of $S$. This implies that frame $F$ must contain at least $k$ non-intersecting edges that correspond to at least $k$ independent vertices of $H$.
\end{proof}

\begin{figure}
    \centering
    $H:$ \includegraphics[page=10]{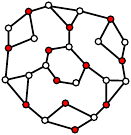} \rule{2cm}{0cm}
    $S:$ \includegraphics[page=11]{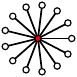}
    \caption{If the size of a maximum independent set of $H$ is $k$ then the size of the maximum frame containing the center of $S$~--~indicated by red vertices~--~is at most $k+1$.
    %If $X$ is the union of a graph $H$ and a star $S$ with $k+1$ leaves then $H$ contains an independent set of size $k$ if and only if $\mu(G) \geq k+1$ for any geometric graph $G$ with crossing graph~$X$. Red vertices in $X$ indicate an initial frame, and in $H$ a maximum independent set.
}
    \label{fig:npCompleteness}
\end{figure}

\section{A Greedy Heuristic}\label{se:heuristic}

We devise variants of a greedy heuristic for the problem \textsc{MaxMinFramePlanarStory($G$)}. 
The heuristic computes the frames one by one. Throughout the algorithm, we call the edges of the current frame \emph{current edges}, the edges that were already removed \emph{past edges}, and those that still have to be inserted \emph{future edges}. The \emph{current/future degree} of an edge $e$ is the number of current/future edges that cross $e$. 
A simple greedy strategy may work as follows.

\medskip

\noindent\textsf{Simple Greedy} 
\vspace{-1ex}
\begin{itemize}
\item{\sf Phase 1.} Start with some initial planar frame~$F_1$, and let $\tau=|E(G)|-|E(F_1)|+1$ be the total number of frames in the story.  
\item{\sf Phase 2.} For each step $i=1, \dots, \tau-1$, let $F_i$ be the current planar frame. Pick a future edge $e$ with the minimum current degree and let $F_{i+1}$ be the frame obtained from $F_i$ by adding $e$ and by removing the edges that are crossed by $e$.
\end{itemize}

\begin{figure}
	\centering
	\begin{minipage}{0.49\linewidth}
		\centering
		\includegraphics[width=0.65\linewidth,page=49]{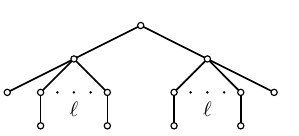}
		\subcaption{Crossing graph with  $4\ell + 5$ vertices\label{fig:crossingGraph}}
	\end{minipage}
	\begin{minipage}{0.49\linewidth}
		\centering
		\includegraphics[width=0.62\linewidth,page=57]{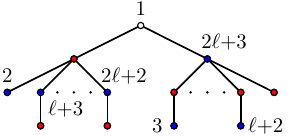}
		\subcaption{$\mu(G)=\lfloor|E(G)|/2\rfloor=2\ell+2$\label{fig:tree_optimum}}
	\end{minipage}\\
	\begin{minipage}{0.55\linewidth}
		\centering
		\includegraphics[width=0.59\linewidth,page=52]{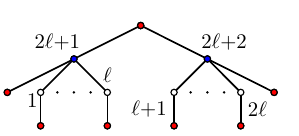}
		\subcaption{\textsf{Simple Greedy} yields minimum frame size $2$\label{fig:tree_greedy}}
	\end{minipage}\hfil
	\begin{minipage}{0.44\linewidth}
		\centering
		\includegraphics[width=0.72\linewidth,page=54]{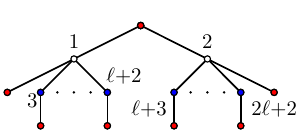}
		\subcaption{$\mu_F(G)=2\ell$\label{fig:tree_better}}
	\end{minipage}\hfil
	\caption{\label{fig:crossinggraph}The crossing graph $X$ of some graph $G$. %Red vertices correspond to the edges of $G$ in the initial frame $F$. 
    The vertex labels indicate the order in which we add the respective edges of $G$ to the initial frame (red vertices).
    The two cases (c) and (d) %\ref{fig:tree_greedy} and \ref{fig:tree_better} 
    in the bottom row start with the same initial frame $F$. 
    }
\end{figure}

Since \textsf{Simple Greedy} always first increases an initial frame to a maximal planar subgraph, 
the example in \cref{fig:notmaximal}
%\cref{le:initial-frame-not-maximal}\todo{SC: This lemma is now in the appendix} 
implies that this strategy is not optimum in general.
%may not achieve the optimum in many cases. 
Moreover, \textsf{Simple Greedy} does not yield a constant-factor approximation even if the crossing graph is a tree; see the example in \cref{fig:crossinggraph}. 
The alternative order in \cref{fig:tree_better}  is
better than the solution of \textsf{Simple Greedy}, as it ensures that the last frame remains sufficiently large. Hence, we introduce a refined version of the simple greedy strategy, described hereunder.
We first describe the general strategy of the two phases of the \textsf{Advanced Greedy} algorithm. In the following we provide details about alternative variants for executing both of them.

\medskip
\noindent\textsf{Advanced Greedy} 
\vspace{-1ex}
\begin{itemize}
    \item{\sf Phase 1.} Compute two edge-disjoint planar frames, the initial frame $F_1$ and the final frame $F_\tau$, where $\tau=|E(G)|-|E(F_1)|+1$.
    \item{\sf Phase 2.} For each step $i=1,\dots,\tau-1$, let $F_i$ be the current planar frame. To make sure that all edges of $F_\tau$ are indeed in the final frame, we say that a future edge is \emph{admissible} if it is %choose the next edge among the future edges that  are either 
    $(i)$ not in $F_\tau$ or $(ii)$ not crossed by any other future edge. 
    %We call this subset of future edges \emph{admissible} for $F_{i+1}$ and 
    We pick an admissible edge $e$ with the minimum current degree. Let $F_{i+1}$ be the frame obtained from $F_i$ by adding $e$ and by removing the edges that are crossed by $e$. 
\end{itemize}

Note that \textsf{Phase~2} of \textsf{Advanced Greedy} is related to the reconfiguration problem~\cite{ito_etal:2011} that asks for a %step-by-step 
transformation between two feasible solutions of a problem such that all intermediate results are also feasible. In particular,~\cite{ito_etal:2020} considered the problem of traversing between two given independent sets of the same size. However, in \cite{ito_etal:2020} at each step exactly one vertex is removed and one vertex is added, and vertices may disappear and~reappear~several~times.

%%%%%%%%%%%%%%% PHASE 1
\subparagraph{Alternative variants for Phase~1.}
We propose three alternative strategies for computing the initial frame $F_1$ and the last frame $F_\tau$, or in other words two possibly large disjoint independent sets $I_1$ (initial frame) and $I_\tau$ (final frame) of the crossing graph $X$ of $G$.

\begin{description}
\item[\sf 1a.] $(I_1,I_\tau)$ is a Pareto optimal maximum pair of independent sets such that $|I_1| \leq |I_\tau|$. See \cref{thm:maximumPairFPT} for the computation.

\item[\sf 1b.] $I_1$ is a large independent set of $X$ of size at most $|E(G)|/2$ and $I_\tau$ is a maximal independent set of $X-I_1$. More precisely, we use the following approach: While $|I_1| \leq |E(G)|/2-1$ and $X-I_1-N(I_1)$ is not empty, iteratively add a vertex of minimum degree in $X-I_1-N(I_1)$ to $I_1$. Then, iteratively add a vertex of minimum degree in $X-I_1-I_\tau-N(I_\tau)$ to $I_\tau$. 

\item[\sf 1c.] %Starting with two empty sets $J_1$ and $J_2$ and alternating between $i=1$ and $i=2$,  iteratively add a vertex of minimum degree in $X-J_{1}-J_2-N(J_i)$ (if there exists one) to $J_i$. If $|J_1| \leq |J_2|$, let $I_1 = J_1$ and $I_\tau = J_2$. Otherwise let $I_1 = J_2$ and $I_\tau = J_1$.
Alternating between $i=1$ and $i=\tau$,  iteratively add a vertex of minimum degree in $X-I_{1}-I_\tau-N(I_i)$ (if there exists one) to $I_i$. If $|I_1| > |I_\tau|$, exchange $I_1$ and $I_\tau$.
\end{description}

Variants~{\sf 1b} and~{\sf 1c} can be implemented in linear time in the size of $X$ by sorting the vertices in buckets according to their degrees. 
Regarding Variant~{\sf 1a},  while it is NP-hard to find a maximum pair of independent sets in a graph 
%ARXIV 
(\cref{thm:balancedHard} in the appendix), 
%\cite{arxiv}, 
there exists a fixed-parameter tractable (FPT) algorithm parameterized by the treewidth of the graph (\Cref{thm:maximumPairFPT}).
Therefore, Variant~{\sf 1a} can be much slower than Variants~{\sf 1b} and~{\sf 1c}, although it may lead to better solutions. We also remark that  \Cref{thm:maximumPairFPT} establishes a result of independent interest, which goes beyond the purposes of our specific problem.

A \emph{tree decomposition}\label{def:treeDecomposition} \cite{RS-84} of a graph $X$ is a tree $\mathcal T$ such that each vertex $\nu \in V(\mathcal T)$ is associated with a set $V(\nu) \subseteq V(X)$, called the \emph{bag} of $\nu$, 
such that %the following conditions hold.
\begin{inparaenum}[(a)]
    \item $\bigcup_{\nu\in V(\mathcal T)}V(\nu)=V(X)$,
    \item for each edge $uv \in E(X)$, there is a vertex $\nu \in V(\mathcal T)$ such that $u,v \in V(\nu)$, and,
    \item for each $v\in V(X)$, the subgraph of $\mathcal T$ induced by the vertices $\{\nu\in V(\mathcal T) \; | \; v \in V(\nu)\}$ is connected.
\end{inparaenum}
The \emph{width} of a tree decomposition is $\max_{\nu \in V(\mathcal T)}|V(\nu)|-1$. 
The \emph{treewidth} tw$(X)$ 
of a graph~$X$ is the minimum among the widths of any of its tree decompositions.

\begin{restatable}
[\restateref{thm:maximumPairFPT}]
{theorem}{maximumPairFPT}
%\begin{theorem}
\label{thm:maximumPairFPT}
    A Pareto optimal maximum pair of independent sets can be computed in $\mathcal O(3^\omega\omega^2n^3)$ time for a graph $X$ with $n$ vertices 
    and a tree decomposition $\mathcal T$ of width~$\omega$. 
%\end{theorem}
\end{restatable}
\begin{proof}[Proof Sketch:]
     The proof follows the ideas of the FPT approach for constructing a minimum vertex cover as proposed by Niedermeier~\cite{niedermeier_book},  enhancing it with the concept of Pareto optimal pairs. 
    Root $\mathcal T$ at an arbitrary vertex. We use dynamic programming.  For a vertex~$\nu$ of~$\mathcal T$ let 
    $X(\nu)$ be the subgraph of~$X$ induced by the vertices in the bags of the subtree of $\mathcal T$ rooted at~$\nu$.
    For each map $c\colon V(\nu)\rightarrow \{$red, blue, white$\}$, we  compute the list $L(\nu,c)$ of Pareto-optimal pairs~$(\alpha,\beta)$ for which $X(\nu)$ contains two disjoint independent sets $I_1$ and $I_2$ of size~$\alpha$ and~$\beta$, respectively, such that the red vertices of~$V(\nu)$ are contained in~$I_1$, the blue vertices of~$V(\nu)$ in~$I_2$, and the white vertices of $V(\nu)$ are neither in $I_1$ nor~$I_2$. 
\end{proof}

\begin{figure}
    \centering
    \begin{minipage}{0.5\linewidth}
	\centering
	\includegraphics[page=80]{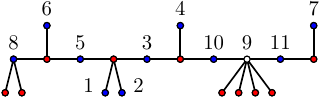}
	\subcaption{greedy: 10,10,11,11,12,12,13,13,12,9,9,10}
    \end{minipage}%
    \begin{minipage}{0.5\linewidth}
	\centering
	\includegraphics[page=79]{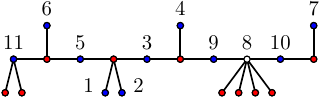}
	\subcaption{optimal: 10,10,11,11,12,12,13,13,10,10,11,10}
    \end{minipage}
    \caption{A Pareto-optimal maximum pair of independent sets in the crossing graph as initial and last frame with a non-optimal greedy story and a min-frame optimal story. %Red vertices indicate the initial frame and t
    The edges are added according to the numbers of the vertices in the crossing graph to the initial frame (red vertices).
    The subcaptions show the sequences of the frame sizes. 
   }
    \label{fig:greedy_notOPt-givenInitialLast}
\end{figure}

Given a Pareto-optimal maximum pair $(I_1,I_\tau)$ of independent sets, \textsf{Advanced Greedy} %, Variant~
{\sf 1a} does not necessarily result in a planar story with minimum frame size 
among all planar stories with initial frame $I_1$ and final frame $I_\tau$. 
An example is given in \cref{fig:greedy_notOPt-givenInitialLast}: here, there is a vertex that is neither in the initial nor in the final frame whose degree is greater than~3.

\begin{restatable}
[\restateref{thm:heurOptCases}]
{lemma}{heurOptCases}
%\begin{theorem}
\label{thm:heurOptCases}
	Starting from a Pareto-optimal maximum pair $(I_1,I_\tau)$ of independent sets,
	\textsf{Advanced Greedy} yields a min-frame optimal planar story for a  geometric graph $G$ if the crossing graph $X$ contains no cycles and the vertices not in $I_1 \cup I_\tau$ have degree at most three.
%\end{theorem}
\end{restatable}

\begin{proof}[Proof Sketch]
    The frame sizes are first monotonically increasing and then monotonically decreasing and,
    thus, at least $\min\{|I_1|,|I_\tau|\}$, which, by \cref{cor:pair},~is~optimum. 	
\end{proof}

\begin{restatable}
[\restateref{thm:2and3planar}]
{theorem}{TwoandThreeplanar}
%\begin{theorem}
\label{thm:2and3planar}
    \textsf{Advanced Greedy} 1a solves \textsc{MaxMinFramePlanarStory($G$)} optimally
    \begin{inparaenum}[(a)]
        \item in linear time for 2-plane geometric graphs, and
        \item in cubic time for those 3-plane geometric graphs for which the crossing graph contains no cycles.
    \end{inparaenum}
%\end{theorem}
\end{restatable}
\begin{proof}[Proof Sketch.]
    The statement for 3-plane graphs is a corollary of \cref{thm:heurOptCases}. For 2-plane graphs, the crossing graph has maximum degree two and, thus, its connected components are paths and cycles, and a maximum pair $(I_1, I_\tau)$ of independent sets can be computed in linear time.  See \cref{fig:2planar}. Moreover, if we apply \textsf{Advanced Greedy} 1a then the minimum frame size is $\min\{|I_1|,|I_\tau|\}$, unless there is an even cycle but no odd path. In this case, the minimum frame size is $\min\{|I_1|,|I_\tau|\}-1$.
\end{proof}

    \begin{figure}
    \centering
    \begin{minipage}[b]{0.19\linewidth}
    \centering
    \includegraphics[page=1]{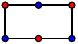}
    \subcaption{even cycles}
    \label{fig:2planar_evenCycle}
    \end{minipage}%
    \begin{minipage}[b]{0.185\linewidth}
    \centering
    \includegraphics[page=2]{2-planar.pdf}
    \subcaption{odd cycles}
    \label{fig:2planar_oddCycle}
    \end{minipage}%
    \begin{minipage}[b]{0.185\linewidth}
    \centering
    \includegraphics[page=3]{2-planar.pdf}
    \subcaption{even paths}
    \label{fig:2planar_evenPath}
    \end{minipage}%
    \begin{minipage}[b]{0.25\linewidth}
    \centering
    \includegraphics[page=5]{2-planar.pdf}
    \subcaption{half of odd paths}
    \label{fig:2planar_oddPath2}
    \end{minipage}%
    \begin{minipage}[b]{0.2\linewidth}
    \centering
    \includegraphics[page=4]{2-planar.pdf}
    \subcaption{other odd paths}
    \label{fig:2planar_oddPath1}
    \end{minipage}%
    \caption{The connected components of the crossing graph for a 2-plane graph. 
    }
    \label{fig:2planar}
\end{figure}

%%%%%%%%%%%%%%% PHASE 2
\subparagraph{Alternative variants for Phase~2.}

In each step of {\sf Phase~2}, the set $E'$ of admissible future edges with minimum current degree might contain multiple edges. For an edge $e \in E'$, denote by $C(e)$ the set of current edges that cross $e$. Also, let $E'' \subseteq E'$ be the subset of edges in $E'$ such that, for any $e \in E''$, the number of future edges crossing some edges in $C(e)$ is the maximum over all other edges in $E'$.
We consider two alternative strategies for selecting the next future edge $e$ that enters the drawing: \textsf{2a.} $e$ is chosen uniformly at random in $E'$; or \textsf{2b.} $e$ is chosen uniformly at random in $E''$. 

Variant \textsf{2b} applies a tie-breaking rule to further restrict the set of admissible edges that can be selected. The rationale is to maximize the amount of future edges that will benefit from the removal of the current neighbors of the edge that will enter the drawing.

\section{An ILP Exact Approach}\label{se:ilp}

We present an exact integer linear program (ILP) for
\textsc{MaxMinFramePlanarStory($G$)}. 
Let $\tau$ be the number of frames necessary, i.e., $\tau \leq |E(G)|$. 
For each edge $e\in E(G)$ and each
$t \in \{1, \dots, \tau\}$ we use a binary variable $x_e^t$ to represent 
that edge $e$ of $G$ is present in frame $t$. 
Additionally, we use \emph{lifting variables} $z_e^t \in \{0, 1\}$, 
where $z_e^t = 1$ indicates that an edge~$e$ of~$G$ appears in frame~$t$.
Finally, we use a variable $y_{\min}$ that
corresponds to the minimum number of edges present over all frames.
\begin{alignat}{3}
    &\text{maximize}\qquad & y_{\min} && \\
    &\text{subject to}    & x_e^t + x_f^t &\leq 1 & \qquad\forall t\in\{1, \dots, \tau\}\;\forall\{e, f\} \in E(X)\label{constr:no-crossing}\\
    &                     & \sum_{t=1}^\tau x_e^t &\geq 1 &\forall e\in E(G) \label{constr:edge-exists-in-frame}\\
    & &\sum_{e\in E(G)} x_e^t &\geq y_{\min}  &\forall t \in \{1, \dots, \tau\}\label{constr:min-objective}\\
    &&\sum_{t=1}^\tau z_e^t &=1 &\forall e \in E(G)\label{constr:appear-once}\\
    & &z_e^t + x_e^{t-1} &\geq x_e^t &\forall e\in E(G)\; \forall t\in\{2, \dots, \tau\}\label{constr:lift1}\\\
    & &z_e^1&\geq x_e^1 &\forall e\in E(G)\label{constr:lift2}\\\
    &&\sum_{e\in E(G)} z_e^t &\leq 1  &\forall t\in\{2, \dots, \tau\}\label{constr:at-most-one-edge-per-frame}\\
    & & x_e^t, z_e^t &\in \{0, 1\} & \\
    %& & z_e^t &\in \{0, 1\} &\\
    & & y_{\min} &\geq 0 &
\end{alignat}

\noindent Constraints~\ref{constr:no-crossing} and~\ref{constr:edge-exists-in-frame}  ensure that no two edges cross in the same frame and that each edge appears in some frames, respectively.
%Constraints~\ref{constr:edge-exists-in-frame} guarantee that each edge
%appears in some frames. 
Constraints~\ref{constr:min-objective}
are a linearization of the objective~function.

Additionally, we need to model the following property. If an edge~$e$ is present in
frame $F_i$ and present in frame $F_j$ with $i < j$, then~$e$ also needs 
to be present in all intermediate frames. 
This can be directly translated, but requires a cubic
number of constraints for each edge. Hence, we introduce 
Constraints~\ref{constr:appear-once}--\ref{constr:lift2} that model the same
property.
Consider an edge~$e$ that appears at frame~$t$ for the first time.
Due to Constraints~\ref{constr:lift1} the lifting variable $z_e^t$ must be
set to 1 if the edge $e$ appears in frame $t$, as $x_e^{t-1} = 0$ or since $t=1$ (Constraints~\ref{constr:lift2}).
The consecutive frame $t+1$ can now include this edge, as $x_e^t=1$ and the
variable $z_e^{t+1}$ is no longer required.
However, once~$e$ disappears in a frame $j>t$, i.e., $x_e^j=0$, this is
no longer possible, as the lifting variables cannot be used anymore due to 
Constraints~\ref{constr:appear-once}.

Finally, we include Constraints~\ref{constr:at-most-one-edge-per-frame} to enforce
that at most one edge appears in each frame, except for the initial frame in which
multiple edges may appear.

\section{Experimental Analysis}\label{se:experiments}
We present the results of an extensive experimental analysis, whose goal is to compare the effectiveness and the efficiency of our heuristics and of the exact algorithm. 
%
%We use the name \textsf{Opt} to denote the exact integer linear program.
For the heuristics, we denote by \textsf{AG-1x2y} the variant of algorithm \textsf{Advanced Greedy} that uses Variant \textsf{1x} for \textsf{Phase 1} and Variant \textsf{2y} for \textsf{Phase 2}, where \textsf{x} can be \textsf{a}, \textsf{b}, or \textsf{c}, while \textsf{y} can be \textsf{a} or \textsf{b}. In total we have 6 possible variants of our heuristic. 
In the following we describe the graph benchmark used for the experiments, the experimental setting, and the obtained results.

\subparagraph{Graph Benchmark.}
To evaluate the performance of our algorithms on different types of graphs, we used several test suites, totaling $637$ instances\footnote{To allow replicability, we made the whole graph benchmark publicly available. However for the sake of anonymity, the URL has been removed from the submitted paper.}. We considered two main types of instances: \textsc{Geometric Graphs} and \textsc{Crossing Graphs}.

\bigskip\noindent\textsc{Geometric Graphs.} A collection of both random and real graphs drawn with a force-directed algorithm. The different graph classes are as follows

\begin{itemize}
    \item \texttt{random graphs:} 
    %$200$ geometric graphs with number of vertices $n \in \{10, 20, \dots, 100\}$ and density (number of edges per number of vertices) $d \in \{1.2, 1.6, 2.0, 2.4\}$. For each pair $(n, d)$ we randomly generated $5$ distinct graphs with $n$ vertices and $m=d \cdot n$ edges, 
    $200$ geometric graphs: For each   $n \in \{10, 20, \dots, 100\}$ and for each density\footnote{The \emph{density} of a graph is the number of its edges divided by the number of its vertices.}  $d \in \{1.2, 1.6, 2.0, 2.4\}$, we randomly generated $5$ distinct graphs with $n$ vertices and $m=d \cdot n$ edges,
    with Erdős--Rényi's model, which guarantees uniform probability distribution. 
    %The process for generating each graph was repeated until we got a non-planar graph. 
    For each generated graph we computed a straight-line drawing through the popular Fruchterman--Reingold force-directed algorithm~\cite{fr-91}. We exploited the NetworkX Python library~\cite{hagberg2008} both for generating the graphs and for computing their drawings. 
    The resulting sizes of the crossing graphs is shown in
    \cref{fig:random_graphs_crossing_graph_sizes}.

    \begin{figure}
        \centering
        \includegraphics{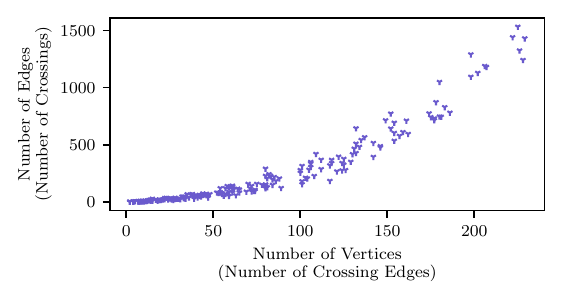}
        \caption{The sizes of the crossing graphs for the \texttt{random graphs}. The $x$-axis is the number of vertices of the crossing graph (i.e., the number of crossing edges of the instance); the $y$-axis is the number of edges of the crossing graph (i.e., the number of crossings in the instance).} 
        \label{fig:random_graphs_crossing_graph_sizes}
    \end{figure}

    \item \texttt{real graphs:} $12$ graphs with 30 to 379 vertices representing data from various real-world domains (see \Cref{tb:real-graphs}).% in \Cref{se:app-experiments}). 
    %In order to enable the comparison with the ILP model, 
    They usually have higher density than the \texttt{random graphs} and, consequently, more edge crossings. All these graphs have been taken from the well-known Network Repository~\cite{nr-aaai15} (\url{https://networkrepository.com}). For each graph, we removed self-loops (if any) and computed a straight-line drawing with the implementation of Fruchterman--Reingold's force-directed algorithm in NetworkX. 
    %The list of graphs is reported in \Cref{tb:real-graphs}.
    \todo[disable]{Test a different drawing algorithm for the real
    graphs, since Fruchterman--Reingold produces drawings with many crossings. 
    $\Rightarrow$ The drawing algorithms available to networkx (and pygraphviz) did not
    produce significantly better drawings than the algorithm by Fruchterman and Reingold.
    }
\end{itemize}

%%%% TABLE REAL-GRAPHS DESCRIPTION
    \begin{table}[htb]
    \centering
    \begin{tabular}{|l|l|r|r|r|}
    \hline
    % header
    {\bf Graph name} & {\bf Category} & $n$ & $m$ & $\chi$\\
    \hline
    % Row 1
    insecta-beetle-group-c1-period-1 & animal social science network& $30$ & $185$ & $1,737$\\
    
    % Row 2
    road-chesapeake & road network & $39$ & $170$ & $1,049$ \\
    
    % Row 3
    eco-stmarks & eco network & $54$ & $353$ & $6,320$\\
    
    % Row 4
    lesmis & miscellaneous network & $77$ & $254$ & $838$ \\

    % Row 5
    ca-sandi\_auths & collaboration network & $86$ & $124$ & $8$\\
    
    % Row 6
    gd06-theory & miscellaneous network & $101$ & $190$ & $1,015$\\

    % Row 7
    polbooks & miscellaneous network & $105$ & $441$ & $2,465$\\
    
    % Row 8
    adjnoun & miscellaneous network & $112$ & $425$ & $6,868$\\

    % Row 9
    rajat11 & miscellaneous network & $135$ & $680$ & $290$\\

    % Row 10
    email-enron-only & email network & $143$ & $623$ & $5,230$\\

    % Row 11
    bwm200 & miscellaneous & $200$ &  $596$ & $7$\\

    % Row 12
    ca-netscience & collaboration network & $379$ &  $914$ & $901$\\
    
    \hline
    \end{tabular}
    \rule{0pt}{1pt} % <-- Vertical space
    \caption{Real graphs with $n$ vertices, $m$ edges, and $\chi$ edge-crossings.
    }
    \label{tb:real-graphs}
\end{table}
 
\bigskip \noindent\textsc{Crossing Graphs.}  %A collection of different types of crossing graphs; namely, w
We generate different types of crossing graphs that are planar. Recall that planar graphs are segment intersection graphs~\cite{chalopinG:stoc09,goncalvesIP:soda18}. Therefore, they always reflect the structure of some geometric graphs in terms of edge crossings. 
    
       \begin{itemize}
            \item \texttt{caterpillars}: $60$ $n$-vertex caterpillar trees, $10$ instances chosen uniformly at random for each value of $n \in \{10, 20, \dots, 60\}$. First, we selected the diameter $\delta$ based on the  probability of an $n$-vertex caterpillar to have such diameter;
            %(i.e., according to the marginal distribution), 
            %then we chose from the corresponding subset by generating a composition of $\delta-1$ terms summing up to $n-\delta-1$~\cite{caterp-enumer} and by rejecting with probability $0.5$ non-symmetric caterpillars (which would be generated twice). 
            then we chose how to randomly attach the $n-\delta-1$ leaves to the $\delta-1$ internal vertices of the caterpillar~\cite{caterp-enumer}. We rejected with probability $0.5$ non-symmetric caterpillars, as they would have been generated twice. 
            Due to the large integers involved in the computation, this method allowed us to generate caterpillars with up to $60$ vertices.
            \item \texttt{trees}: $80$ general trees, with %number of vertices 
            $n \in \{10, 20, \dots, 80\}$ vertices; for each fixed $n$ we collected $10$ distinct instances, generated uniformly at random with the algorithm in~\cite{DBLP:journals/algorithmica/AlonsoRS97} 
            %[5-45] vertices $40$ for each sample
            \item \texttt{series-parallel graphs}: $105$ series-parallel graphs, with %number of vertices 
            $n \in \{10, 20, \dots, 70\}$ vertices and density $d \in \{1.2, 1.4, 1.6\}$. For each pair $(n, d)$ we collected $5$ instances, 
            %The generative procedure consisted in 
            by repeatedly generating biconnected planar graphs 
            using 
            %with the algorithm available in the latest version of 
            the OGDF library~\cite{DBLP:reference/crc/ChimaniGJKKM13} (\url{https://ogdf.uos.de/}), and discarding those graphs that were not series-parallel. %(i.e., whose SPQR-tree contained some rigid nodes).
            %%%%%% Otherwsie we'd have to define SPQR-trees
            \item \texttt{planar graphs}: $180$ general connected planar graphs, with $n \in \{10, 20, \dots, 90\}$ and density $d \in \{1.2, 1.6, 2.0, 2.4\}$. Precisely, we generated $5$ instances for each pair $(n, d)$, still using the random generator available in the OGDF library~\cite{DBLP:reference/crc/ChimaniGJKKM13}.
       \end{itemize}

       Since \texttt{trees} and \texttt{caterpillars} have treewidth $1$ and the treewidth of \texttt{series-parallel graphs} is $2$, these instances 
       %in these sets 
       are of particular interest for evaluating the performances of the heuristics \textsf{AG-1a2y}, which works in polynomial-time for crossing graphs of bounded treewidth.

\subparagraph{Experimental setting.}
We performed the experiments on a Windows 11 Pro machine with an Intel Core i5-12500 CPU and 16 GB RAM. We used Gurobi \cite{gurobi} version 12.0.1 as ILP solver. We implemented the heuristics in Python, using the NetworkX library~\cite{hagberg2008}.

On those instances for which the exact algorithm was able to finish the computation within the given time (see below), we compared the optimum with the value computed by the various heuristics. Precisely, we measured the ratio between the minimum frame size of the story computed by each heuristic and the minimum frame size of the story computed by the exact algorithm. With a slight abuse of terminology, we call this value the \emph{approximation ratio}.
For the \texttt{random graphs} and \texttt{real graphs}, the computation of the approximation ratio is restricted to the subset of crossing edges, because crossing-free edges persist in all frames of a planar graph story, regardless of the algorithm used to compute it. 

For each instance, we also measured the runtime of each algorithm. %In this regard, 
We fixed an upper bound of $T = 15$ minutes for the runtime allowed for the exact algorithm, and we recorded the value of the best solution  achieved within~$T$ (i.e., the incumbent solution). For those instances for which the runtime of the exact algorithm exceeded $T$, we also measured the optimality gap, that is, the \emph{gap} between the incumbent solution (IS) and the upper bound (UB) estimated by the ILP solver at time $T$. In formula: $gap = \frac{UB-IS}{UB}$.  

Furthermore, after verifying that, for all instances optimally solved by the exact algorithm, the heuristics were also able to produce a solution in less time, we set a maximum time limit of $T' = 10$ minutes for the heuristics based on Variant \textsf{1a}, whose runtime may grow exponentially in the treewidth of the crossing graph. This was done to allow for faster experimentation, given the large number of instances. The instances for which Variant \textsf{1a} could not finish within $T'$ were considered unfeasible for \textsf{AG-1a2a} and \textsf{AG-1a2b}.

\subparagraph{Performance of the exact algorithm.} 
\Cref{fig:runtime_ilp} shows the percentage of instances per graph class that the ILP was able to solve within the given time limit of $T=15$ minutes ($900$ seconds). In each chart, the $x$-axis is the 
time in seconds and the $y$-axis (blue line) reports the number of instances solved optimally within the time at the corresponding $x$ coordinate. 
\Cref{fig:ilp-gaps} shows the optimality gap across the instances for each graph class (a point for each instance), where the instances were sorted by the number of crossings; points of instances with the same gap coincide in the chart. 

Regarding the \textsc{Geometric Graphs} benchmark, more than half (55\%) of the {\tt random graphs} were solved optimally. 
In particular, all instances with 90 or less crossing edges were solved optimally in the given time, and 90\% of the instances with 150 or less crossing edges were solved optimally.
Conversely, only the two {\tt real graphs} with small number of crossings (namely, the graphs bmw200 and ca-sandi\_auths) could be solved optimally within~$T$. However, for the graph rajat11 (which contains 290 crossings), the optimality gap is close to $0$ (namely, $0.08$). Overall, for most of the instances in the {\tt random graphs} and {\tt real graphs}, the optimality gap tends to increase with the number of edge crossings.

\begin{figure}[tb]
    \centering
    \includegraphics[width=1\columnwidth]{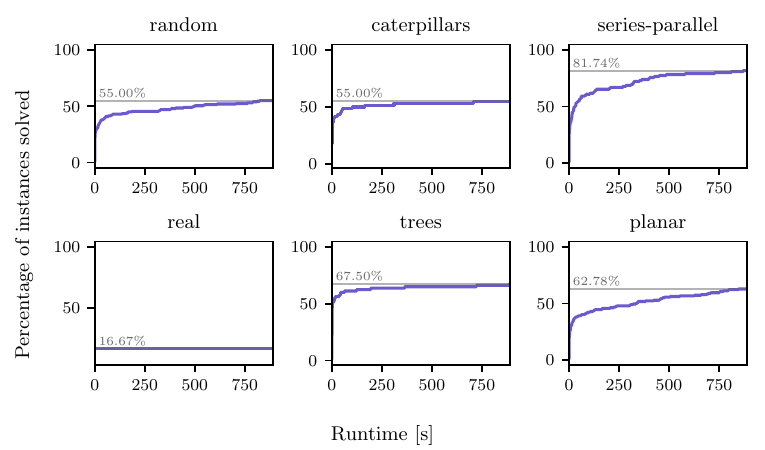}
    \caption{Percentage of instances solved (optimally) by the exact algorithm
    within time $T$.}
    \label{fig:runtime_ilp}
\end{figure}

Regarding the \textsc{Crossing Graphs} benchmark, most of the {\tt series-parallel graphs} (81.74\%) and a high percentage of the {\tt trees} (67.5\%) and {\tt planar graphs} (62.78\%) could be solved optimally. 
Interestingly, more randomly sampled trees were solved than caterpillars (55.00\%) in the given time. 
However, it is worth noting that the regular pattern observed in the optimality gaps for the {\tt caterpillars} and {\tt trees} is due to the fact that, for instances not solved to optimality within the time limit, the absolute difference between the best bound and the incumbent solution was exactly 1.

\begin{figure}[h!]
    \centering
    \includegraphics[width=1\columnwidth]{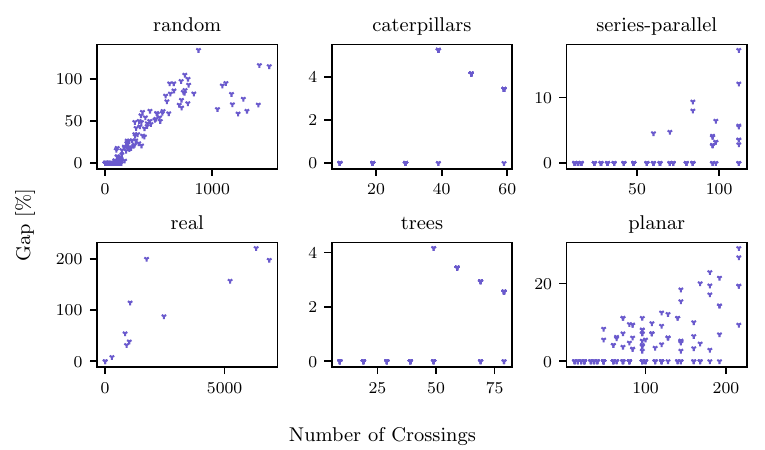}
    \caption{The optimality gap (\%) reached by the ILP solver in the given time for each graph class. The instances of each graph class is sorted in ascending order by the number of crossings.}
    \label{fig:ilp-gaps}
\end{figure}

\subparagraph{Performance of the heuristics.}
The runtime of all the heuristics is reported in \Cref{se:app-experiments} (see \cref{tb:running-time-real-graphs} for the {\tt real graphs} and \cref{fig:time_plots} for the others).
%Full details on the runtime of all the heuristics can be found in the long version of the paper~\cite{arxiv}.
Heuristics \textsf{AG-1b2a}, \textsf{AG-1b2b}, \textsf{AG-1c2a}, and \textsf{AG-1c2b} computed all the instances of our graph benchmark very efficiently. On average they took less than $0.1$ seconds on the {\tt random graphs} and about $0.34$ seconds on {\tt real graphs} (the most expensive required about $1.6$ seconds). Also, they took less than $10$ milliseconds for the instances in the \textsc{Crossing Graphs} datasets.
Runtime differences among these heuristics are negligible across Variants~\textsf{2a} and~\textsf{2b} for \textsf{Phase 2}, although \textsf{2a} is slightly faster than \textsf{2b}. The runtime is mostly affected by the variant chosen for \textsf{Phase 1}, where Variant~\textsf{1b} is a bit faster than Variant~\textsf{1c}, as~expected.

\begin{figure}[tb]
	\begin{subfigure}[b]{.9\linewidth}
		\centering 
		\includegraphics[width=\linewidth,page=1]{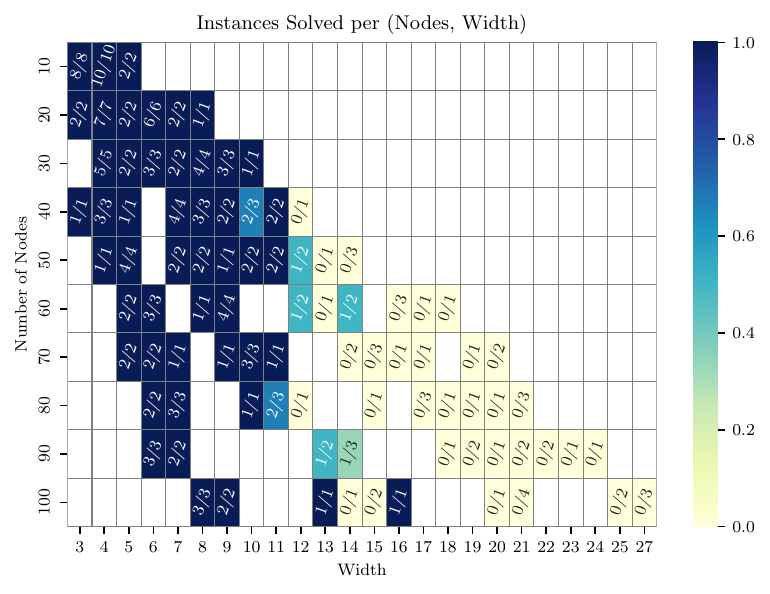} 
	\end{subfigure}
	\caption{\label{fig:solved_width} Fraction of the instances in the {\tt random graphs} dataset terminated by Variant {\tt 1a} within~$T'$, with respect to the number of graph nodes and the width of the tree decomposition.}
\end{figure}

Conversely, Variant~\textsf{1a} is significantly more expensive. We implemented it, using the Minimum Fill-in heuristic \cite{bodlaender_treewidth} to compute tree decompositions. For the analysis of the heuristics based on this variant, we took into account only those instances where a Pareto-optimal maximum pair of independent sets was found within the given time $T'=10$ minutes. Variant~\textsf{1a} was able to find a solution for all the instances in the \textsc{Crossing Graphs} benchmark, as the Minimum Fill-in heuristic produced tree decompositions of width at most six for planar graphs and at most two for series-parallel graphs. For the {\tt random graphs}, Variant~\textsf{1a} solved $137$ (of the $200$) instances within time $T'$; the fraction of instances solved for each combination of the two parameters ``number of nodes'' and ``width of the tree decomposition'' is illustrated in the heatmap of \cref{fig:solved_width}. 
For the {\tt real graphs}, Variant~\textsf{1a} solved only 3/12 instances within $T'$ (see \cref{tb:running-time-real-graphs} in \Cref{se:app-experiments}).

\begin{table}[b]
  \centering
  \begin{tabular}{|l || p{0.17\linewidth}| p{0.19\linewidth} | p{0.22\linewidth}| p{0.19\linewidth}|}
  	\hline
  	& \multicolumn{4}{r|}{\textbf{All Graphs Combined}}\\
  	\hline
  	& \textsf{\% Opt} & \textsf{Avg approx. ratio} & \textsf{Min--Max approx. ratio} & \textsf{SD approx. ratio} \\
  	\hline
  	\textsf{AG-1a2a} & $59.9\%$ (243/406) & $94.14\%$ & $50.00\%$ -- $100.00\%$ & $9.87\%$ \\
  	\hline
  	\textsf{AG-1a2b} & $63.1\%$ (256/406) & $94.72\%$ & $50.00\%$ -- $100.00\%$ & $9.37\%$ \\
  	\hline
  	\textsf{AG-1b2a} & $5.7\%$ (23/406) & $67.91\%$ & $33.33\%$ -- $100.00\%$ & $13.44\%$ \\
  	\hline
  	\textsf{AG-1b2b} & $5.4\%$ (22/406) & $68.21\%$ & $42.86\%$ -- $100.00\%$ & $13.07\%$ \\
  	\hline
  	\textsf{AG-1c2a} & $14.3\%$ (58/406) & $82.76\%$ & $44.44\%$ -- $100.00\%$ & $10.36\%$ \\
  	\hline
  	\textsf{AG-1c2b} & $15.3\%$ (62/406) & $83.45\%$ & $50.00\%$ -- $100.00\%$ & $10.35\%$ \\
  	\hline
  \end{tabular}
  \caption{The performance of our heuristics in comparison to the exact algorithm.}
  \label{tab:heuristic_performance}
\end{table}

\bigskip
Regarding the effectiveness of our heuristics, \Cref{tab:heuristic_performance} summarizes the comparison over all instances that the exact algorithm managed within the given time $T=15$ minutes. 
Detailed results for the different types of graphs are reported in \Cref{se:app-experiments} (\Cref{tab:heuristic_performance_separate}).
%Detailed results for the different types of graphs are reported in~\cite{arxiv}.
Each table reports: the percentage of instances solved optimally by the heuristics (\textsf{\%~Opt}); the approximation ratio (in percentage) averaged over all instances (\textsf{Avg approx. ratio}); the minimum and maximum approximation ratio over all instances (\textsf{Min--Max approx. ratio}), and the standard deviation of the approximation ratio (\textsf{SD approx. ratio}).   

The results show that \textsf{AG-1a2a} and \textsf{AG-1a2b} achieve a $\thicksim 60\%$ optimality rate and an average approximation ratio of $\thicksim 95\%$, significantly outperforming the alternatives \textsf{AG-1b2a/AG-1b2b} and \textsf{AG-1c2a/AG-1c2b} in both metrics. The standard deviation of the approximation ratio for \textsf{AG-1a2a} and \textsf{AG-1a2b} remains below $10\%$, indicating consistent performance. Notably, for the {\tt trees} and the {\tt caterpillars}, \textsf{AG-1a2a} and \textsf{AG-1a2b} achieve near-perfect optimality rates ($97.0\%$ and $100.0\%$, respectively) with approximation ratios close to $100\%$. In contrast, the heuristics based on Variant~\textsf{1b} achieve fewer than $6\%$ optimal solutions on average, with lower approximation ratios around $68\%$. The heuristics using Variant~\textsf{1c}, while also exhibiting low optimality rates ($\thicksim 15\%$), maintain a stronger average approximation ratio of $\thicksim 83\%$ and a good standard deviation ($\thicksim 10.4\%$). 

We remark that, for the {\tt real graphs} the exact algorithm managed only 2 instances within the given time $T$, and \textsf{AG-1a2a} and \textsf{AG-1a2b} managed only 1 additional instance within the given time $T'$. This data is not enough to get a clear picture of how the other heuristics compare on this dataset. However, by examining the data for all the instances in the {\tt real graphs} dataset, we confirm that also in this case the heuristics based on Variant~\textsf{1c} outperformed those based on Variant~\textsf{1b}, with an improvement of about $29\%$ on average.

\subparagraph{Main findings.} The exact algorithm typically handles geometric graphs with up to $90-100$ crossing edges in few minutes. However, it becomes generally impractical for real-world graphs that are locally dense and inevitably contain many more crossing edges. 
Regarding the heuristics, the difference between Variants~\textsf{2a} and~\textsf{2b}, in terms of both efficiency and effectiveness, is often negligible. 
A more detailed analysis of this phenomenon is given in \Cref{se:app-experiments}. 
%A more detailed analysis of this phenomenon is given in~\cite{arxiv}. 
Conversely, the choice of the variant for \textsf{Phase~1} makes a big~difference, namely:

\begin{itemize}
        \item  The heuristics based on Variant~\textsf{1a} yield solutions close to the optimum in many cases. However, they become computationally unfeasible for graphs with high number of nodes or with crossing graphs of high treewidth; in particular, they exhibit more or less the same limits as the exact algorithm on the \textsf{real graphs}, although they make it possible to manage all instances in the \textsc{Crossing Graphs} benchmark.
        
        \item The heuristics based on Variant~\textsf{1c} offer the best trade-off between efficiency and effectiveness. They perform fast on all instances and maintain high average approximation ratios (approximately $80\% - 88\%$, depending on the graph~class). 
\end{itemize}

\section{Final Remarks and Future Work}\label{se:conclusions}

We have proposed the \textsc{MaxMinFramePlanarStory} model, where we visualize a graph in a sequence of planar frames such that only one edge at the time enters the story and the minimum frame size is maximized.  We suggest the following research directions.

\subparagraph{RD1.} In addition to ensuring drawing stability and low visual complexity, our model promotes smooth transitions between consecutive frames (i.e., minimal local changes from one frame to the next), as only the current edges that intersect the upcoming edge are removed.
However, this approach may result in long stories. To balance these aspects, we can compress a planar story by allowing multiple edges to be inserted simultaneously: for each subsequence of consecutive, pairwise non-crossing edges, all of them can be added at once. In this way, the resulting compressed story remains planar, and its smallest frame is at least as large as that of the original story.
Nonetheless, enabling multiple edges to enter simultaneously in the visualization motivates the design of dedicated algorithms aimed at producing short stories while still maximizing the minimum frame size.

\subparagraph{RD2.} We focused on constructing a planar story of geometric graphs, i.e., of graphs with a given drawing. It would be interesting to investigate how different drawings of the same graph affect the quality of the story. 
Furthermore, what kind of crossing graphs would allow us to construct good edge stories?    

\subparagraph{RD3.} Our heuristics based on Variants \textsf{1b} and \textsf{1c} are fast enough to solve also more complex instances than those in our benchmark. It would be interesting to investigate whether the quality of these variants could be improved by using different heuristics for choosing the independent sets for the initial and final frames than just adding vertices of minimum degree. See, e.g., \cite{grossmannLammSchulzStrash:2024} for a  list of heuristics for computing large independent sets.

\subparagraph{RD4.} Regarding the problem complexity: Is
\textsc{MaxMinFramePlanarStoryD}($G$,$m$) NP-complete for $3$-plane graphs? Is the problem in FPT parameterized by the treewidth of the crossing graph of $G$? Is \textsf{Advanced Greedy} \textsf{1a} a constant-factor approximation algorithm?

%\clearpage
\bibliography{references}

\begin{thebibliography}{10}

\bibitem{DBLP:journals/algorithmica/AlonsoRS97}
Laurent Alonso, Jean{-}Luc R{\'{e}}my, and Ren{\'{e}} Schott.
\newblock A linear-time algorithm for the generation of trees.
\newblock {\em Algorithmica}, 17(2):162--183, 1997.
\newblock \href {https://doi.org/10.1007/BF02522824}
  {\path{doi:10.1007/BF02522824}}.

\bibitem{DBLP:conf/gd/ArchambaultLNPT24}
Daniel Archambault, Giuseppe Liotta, Martin N{\"{o}}llenburg, Tommaso Piselli,
  Alessandra Tappini, and Markus Wallinger.
\newblock Bundling-aware graph drawing.
\newblock In Stefan Felsner and Karsten Klein, editors, {\em 32nd International
  Symposium on Graph Drawing and Network Visualization ({GD} 2024)}, volume 320
  of {\em LIPIcs}, pages 15:1--15:19. Schloss Dagstuhl - Leibniz-Zentrum
  f{\"{u}}r Informatik, 2024.
\newblock \href {https://doi.org/10.4230/LIPICS.GD.2024.15}
  {\path{doi:10.4230/LIPICS.GD.2024.15}}.

\bibitem{balaban:socg95}
Ivan~J. Balaban.
\newblock An optimal algorithm for finding segments intersections.
\newblock In {\em Proceedings of the 11th Annual Symposium on Computational
  Geometry (SoCG'95)}, page 211–219. Association for Computing Machinery,
  1995.
\newblock \href {https://doi.org/10.1145/220279.220302}
  {\path{doi:10.1145/220279.220302}}.

\bibitem{DBLP:journals/cgf/BeckBDW17}
Fabian Beck, Michael Burch, Stephan Diehl, and Daniel Weiskopf.
\newblock A taxonomy and survey of dynamic graph visualization.
\newblock {\em Comput. Graph. Forum}, 36(1):133--159, 2017.
\newblock \href {https://doi.org/10.1111/CGF.12791}
  {\path{doi:10.1111/CGF.12791}}.

\bibitem{DBLP:conf/gd/BinucciGLLMNS22}
Carla Binucci, Emilio {Di Giacomo}, William~J. Lenhart, Giuseppe Liotta,
  Fabrizio Montecchiani, Martin N{\"{o}}llenburg, and Antonios Symvonis.
\newblock On the complexity of the storyplan problem.
\newblock In Patrizio Angelini and Reinhard von Hanxleden, editors, {\em 30th
  International Symposium Graph Drawing and Network Visualization ({GD}'22)},
  volume 13764 of {\em Lecture Notes in Computer Science}, pages 304--318.
  Springer, 2022.
\newblock \href {https://doi.org/10.1007/978-3-031-22203-0\_22}
  {\path{doi:10.1007/978-3-031-22203-0\_22}}.

\bibitem{DBLP:journals/jcss/BinucciGLLMNS24}
Carla Binucci, Emilio {Di Giacomo}, William~J. Lenhart, Giuseppe Liotta,
  Fabrizio Montecchiani, Martin N{\"{o}}llenburg, and Antonios Symvonis.
\newblock On the complexity of the storyplan problem.
\newblock {\em J. Comput. Syst. Sci.}, 139:103466, 2024.
\newblock \href {https://doi.org/10.1016/J.JCSS.2023.103466}
  {\path{doi:10.1016/J.JCSS.2023.103466}}.

\bibitem{bodlaender_treewidth}
Hans~L. Bodlaender.
\newblock Discovering treewidth.
\newblock In Peter Vojt{\'a}{\v{s}}, M{\'a}ria Bielikov{\'a}, Bernadette
  Charron-Bost, and Ondrej S{\'y}kora, editors, {\em SOFSEM 2005: Theory and
  Practice of Computer Science}, pages 1--16, Berlin, Heidelberg, 2005.
  Springer Berlin Heidelberg.
\newblock \href {https://doi.org/10.1007/978-3-540-30577-4_1}
  {\path{doi:10.1007/978-3-540-30577-4_1}}.

\bibitem{DBLP:journals/jgaa/BorrazzoLBFP20}
Manuel Borrazzo, Giordano {Da Lozzo}, Giuseppe {Di Battista}, Fabrizio Frati,
  and Maurizio Patrignani.
\newblock Graph stories in small area.
\newblock {\em J. Graph Algorithms Appl.}, 24(3):269--292, 2020.
\newblock \href {https://doi.org/10.7155/JGAA.00530}
  {\path{doi:10.7155/JGAA.00530}}.

\bibitem{chalopinG:stoc09}
J{\'{e}}r{\'{e}}mie Chalopin and Daniel Gon{\c{c}}alves.
\newblock Every planar graph is the intersection graph of segments in the
  plane: extended abstract.
\newblock In Michael Mitzenmacher, editor, {\em Proceedings of the 41st Annual
  {ACM} Symposium on Theory of Computing, {STOC} 2009}, pages 631--638. {ACM},
  2009.
\newblock \href {https://doi.org/10.1145/1536414.1536500}
  {\path{doi:10.1145/1536414.1536500}}.

\bibitem{DBLP:reference/crc/ChimaniGJKKM13}
Markus Chimani, Carsten Gutwenger, Michael J{\"{u}}nger, Gunnar~W. Klau,
  Karsten Klein, and Petra Mutzel.
\newblock The open graph drawing framework {(OGDF)}.
\newblock In Roberto Tamassia, editor, {\em Handbook on Graph Drawing and
  Visualization}, pages 543--569. Chapman and Hall/CRC, 2013.

\bibitem{darmann/doecker:20}
Andreas Darmann and Janosch Döcker.
\newblock On a simple hard variant of {Not-All-Equal} 3-{Sat}.
\newblock {\em Theoretical Computer Science}, 815:147--152, 2020.
\newblock \href {https://doi.org/10.1016/j.tcs.2020.02.010}
  {\path{doi:10.1016/j.tcs.2020.02.010}}.

\bibitem{DBLP:conf/gd/BattistaDGGOPT22}
Giuseppe {Di Battista}, Walter Didimo, Luca Grilli, Fabrizio Grosso, Giacomo
  Ortali, Maurizio Patrignani, and Alessandra Tappini.
\newblock Small point-sets supporting graph stories.
\newblock In Patrizio Angelini and Reinhard von Hanxleden, editors, {\em 30th
  International Symposium on Graph Drawing and Network Visualization ({GD}
  2022)}, volume 13764 of {\em Lecture Notes in Computer Science}, pages
  289--303. Springer, 2022.
\newblock \href {https://doi.org/10.1007/978-3-031-22203-0\_21}
  {\path{doi:10.1007/978-3-031-22203-0\_21}}.

\bibitem{DBLP:journals/jgaa/BattistaDGGOPT23}
Giuseppe {Di Battista}, Walter Didimo, Luca Grilli, Fabrizio Grosso, Giacomo
  Ortali, Maurizio Patrignani, and Alessandra Tappini.
\newblock Small point-sets supporting graph stories.
\newblock {\em J. Graph Algorithms Appl.}, 27(8):651--677, 2023.
\newblock \href {https://doi.org/10.7155/JGAA.00639}
  {\path{doi:10.7155/JGAA.00639}}.

\bibitem{DBLP:conf/gd/GiacomoDLMT13}
Emilio {Di Giacomo}, Walter Didimo, Giuseppe Liotta, Fabrizio Montecchiani, and
  Ioannis~G. Tollis.
\newblock Exploring complex drawings via edge stratification.
\newblock volume 8242 of {\em Lecture Notes in Computer Science}, pages
  304--315. Springer, 2013.
\newblock \href {https://doi.org/10.1007/978-3-319-03841-4\_27}
  {\path{doi:10.1007/978-3-319-03841-4\_27}}.

\bibitem{DBLP:journals/vlc/GiacomoDLMT14}
Emilio {Di Giacomo}, Walter Didimo, Giuseppe Liotta, Fabrizio Montecchiani, and
  Ioannis~G. Tollis.
\newblock Techniques for edge stratification of complex graph drawings.
\newblock {\em J. Vis. Lang. Comput.}, 25(4):533--543, 2014.
\newblock \href {https://doi.org/10.1016/J.JVLC.2014.05.001}
  {\path{doi:10.1016/J.JVLC.2014.05.001}}.

\bibitem{Eades}
Peter Eades, Quan~Hoang Nguyen, and Seok{-}Hee Hong.
\newblock Drawing big graphs using spectral sparsification.
\newblock In Fabrizio Frati and Kwan{-}Liu Ma, editors, {\em 25th International
  Symposium on Graph Drawing and Network Visualization ({GD} 2017)}, volume
  10692 of {\em Lecture Notes in Computer Science}, pages 272--286. Springer,
  2017.
\newblock \href {https://doi.org/10.1007/978-3-319-73915-1\_22}
  {\path{doi:10.1007/978-3-319-73915-1\_22}}.

\bibitem{DBLP:conf/sofsem/FialaFLWZ24}
Jir{\'{\i}} Fiala, Oksana Firman, Giuseppe Liotta, Alexander Wolff, and
  Johannes Zink.
\newblock Outerplanar and forest storyplans.
\newblock In Henning Fernau, Serge Gaspers, and Ralf Klasing, editors, {\em
  49th International Conference on Current Trends in Theory and Practice of
  Computer Science ({SOFSEM} 2024)}, volume 14519 of {\em Lecture Notes in
  Computer Science}, pages 211--225. Springer, 2024.
\newblock \href {https://doi.org/10.1007/978-3-031-52113-3\_15}
  {\path{doi:10.1007/978-3-031-52113-3\_15}}.

\bibitem{fr-91}
Thomas M.~J. Fruchterman and Edward~M. Reingold.
\newblock Graph drawing by force-directed placement.
\newblock {\em Software: Practice and Experience}, 21(11):1129--1164, 1991.
\newblock \href {https://doi.org/10.1002/spe.4380211102}
  {\path{doi:10.1002/spe.4380211102}}.

\bibitem{goncalvesIP:soda18}
Daniel Gon{\c{c}}alves, Lucas Isenmann, and Claire Pennarun.
\newblock Planar graphs as {L}-intersection or {L}-contact graphs.
\newblock In Artur Czumaj, editor, {\em Proceedings of the Twenty-Ninth Annual
  {ACM-SIAM} Symposium on Discrete Algorithms ({SODA} 2018)}, pages 172--184.
  {SIAM}, 2018.
\newblock \href {https://doi.org/10.1137/1.9781611975031.12}
  {\path{doi:10.1137/1.9781611975031.12}}.

\bibitem{grossmannLammSchulzStrash:2024}
Ernestine Gro{\ss}mann, Sebastian Lamm, Christian Schulz, and Darren Strash.
\newblock Finding near-optimal weight independent sets at scale.
\newblock {\em Journal of Graph Algorithms and Applications}, 28(1):439–473,
  2024.
\newblock \href {https://doi.org/10.7155/jgaa.v28i1.2997}
  {\path{doi:10.7155/jgaa.v28i1.2997}}.

\bibitem{gurobi}
{Gurobi Optimization, LLC}.
\newblock {Gurobi Optimizer Reference Manual}, 2024.
\newblock URL: \url{https://www.gurobi.com}.

\bibitem{hagberg2008}
Aric~A Hagberg, Daniel~A Schult, and Pieter~J Swart.
\newblock Exploring network structure, dynamics, and function using networkx.
\newblock {\em Proceedings of the 7th Python in Science Conferences (SciPy
  2008)}, pages 11--15, 2008.

\bibitem{DBLP:journals/cgf/HoltenW09}
Danny Holten and Jarke~J. van Wijk.
\newblock Force-directed edge bundling for graph visualization.
\newblock {\em Comput. Graph. Forum}, 28(3):983--990, 2009.
\newblock \href {https://doi.org/10.1111/J.1467-8659.2009.01450.X}
  {\path{doi:10.1111/J.1467-8659.2009.01450.X}}.

\bibitem{BCproxy}
Seok{-}Hee Hong, Quan~Hoang Nguyen, Amyra Meidiana, Jiaxi Li, and Peter Eades.
\newblock {BC} tree-based proxy graphs for visualization of big graphs.
\newblock In {\em {IEEE} PacificVis 2018}, pages 11--20, 2018.
\newblock \href {https://doi.org/10.1109/PACIFICVIS.2018.00011}
  {\path{doi:10.1109/PACIFICVIS.2018.00011}}.

\bibitem{ito_etal:2011}
Takehiro Ito, Erik~D. Demaine, Nicholas~J.A. Harvey, Christos~H. Papadimitriou,
  Martha Sideri, Ryuhei Uehara, and Yushi Uno.
\newblock On the complexity of reconfiguration problems.
\newblock {\em Theoretical Computer Science}, 412(12):1054--1065, 2011.
\newblock \href {https://doi.org/10.1016/j.tcs.2010.12.005}
  {\path{doi:10.1016/j.tcs.2010.12.005}}.

\bibitem{ito_etal:2020}
Takehiro Ito, Marcin Kamiński, Hirotaka Ono, Akira Suzuki, Ryuhei Uehara, and
  Katsuhisa Yamanaka.
\newblock Parameterized complexity of independent set reconfiguration problems.
\newblock {\em Discrete Applied Mathematics}, 283:336--345, 2020.
\newblock \href {https://doi.org/10.1016/j.dam.2020.01.022}
  {\path{doi:10.1016/j.dam.2020.01.022}}.

\bibitem{caterp-enumer}
Yosuke Kikuchi, Hiroyuki Tanaka, Shin-ichi Nakano, and Yukio Shibata.
\newblock How to obtain the complete list of caterpillars.
\newblock In Tandy Warnow and Binhai Zhu, editors, {\em Computing and
  Combinatorics}, pages 329--338, Berlin, Heidelberg, 2003. Springer Berlin
  Heidelberg.

\bibitem{DBLP:reference/crc/Kobourov13}
Stephen~G. Kobourov.
\newblock Force-directed drawing algorithms.
\newblock In Roberto Tamassia, editor, {\em Handbook on Graph Drawing and
  Visualization}, pages 383--408. Chapman and Hall/CRC, 2013.

\bibitem{kratochvilM:jct94}
J.~Kratochv\'{\i}l and J.~Matousek.
\newblock Intersection graphs of segments.
\newblock {\em Journal of Combinatorial Theory, Series B}, 62(2):289--315,
  1994.
\newblock \href {https://doi.org/10.1006/jctb.1994.1071}
  {\path{doi:10.1006/jctb.1994.1071}}.

\bibitem{m-fcgpa-01}
Bojan Mohar.
\newblock Face covers and the genus problem for apex graphs.
\newblock {\em Journal of Combinatorial Theory, Series B}, 82(1):102--117,
  2001.
\newblock \href {https://doi.org/10.1006/jctb.2000.2026}
  {\path{doi:10.1006/jctb.2000.2026}}.

\bibitem{moret:88}
Bernard M.~B. Moret.
\newblock Planar {NAE3SAT} is in {P}.
\newblock {\em SIGACT News}, 19(2):51–54, 1988.
\newblock \href {https://doi.org/10.1145/49097.49099}
  {\path{doi:10.1145/49097.49099}}.

\bibitem{Nguyen}
Quan~Hoang Nguyen, Seok{-}Hee Hong, and Peter Eades.
\newblock {TGI-EB:} {A} new framework for edge bundling integrating topology,
  geometry and importance.
\newblock In Marc~J. van Kreveld and Bettina Speckmann, editors, {\em 19th
  International Symposium on Graph Drawing ({GD} 2011)}, volume 7034 of {\em
  Lecture Notes in Computer Science}, pages 123--135. Springer, 2011.
\newblock \href {https://doi.org/10.1007/978-3-642-25878-7\_13}
  {\path{doi:10.1007/978-3-642-25878-7\_13}}.

\bibitem{proxy}
Quan~Hoang Nguyen, Seok{-}Hee Hong, Peter Eades, and Amyra Meidiana.
\newblock Proxy graph: Visual quality metrics of big graph sampling.
\newblock {\em {IEEE} Trans. Vis. Comput. Graph.}, 23(6):1600--1611, 2017.
\newblock \href {https://doi.org/10.1109/TVCG.2017.2674999}
  {\path{doi:10.1109/TVCG.2017.2674999}}.

\bibitem{niedermeier_book}
Rolf Niedermeier.
\newblock {\em Invitation to Fixed-Parameter Algorithms}, volume~31 of {\em
  Oxford Lecture Series in Mathematics and its Applications}.
\newblock Oxford University Press, 2006.

\bibitem{DBLP:conf/gd/Purchase97}
Helen~C. Purchase.
\newblock Which aesthetic has the greatest effect on human understanding?
\newblock In Giuseppe {Di Battista}, editor, {\em 5th International Symposium
  on Graph Drawing ({GD} '97)}, volume 1353 of {\em Lecture Notes in Computer
  Science}, pages 248--261. Springer, 1997.
\newblock \href {https://doi.org/10.1007/3-540-63938-1\_67}
  {\path{doi:10.1007/3-540-63938-1\_67}}.

\bibitem{DBLP:journals/iwc/Purchase00}
Helen~C. Purchase.
\newblock Effective information visualisation: a study of graph drawing
  aesthetics and algorithms.
\newblock {\em Interact. Comput.}, 13(2):147--162, 2000.
\newblock \href {https://doi.org/10.1016/S0953-5438(00)00032-1}
  {\path{doi:10.1016/S0953-5438(00)00032-1}}.

\bibitem{DBLP:journals/ese/PurchaseCA02}
Helen~C. Purchase, David~A. Carrington, and Jo{-}Anne Allder.
\newblock Empirical evaluation of aesthetics-based graph layout.
\newblock {\em Empir. Softw. Eng.}, 7(3):233--255, 2002.

\bibitem{DBLP:journals/tvcg/PurchasePP12}
Helen~C. Purchase, Christopher Pilcher, and Beryl Plimmer.
\newblock Graph drawing aesthetics - created by users, not algorithms.
\newblock {\em {IEEE} Trans. Vis. Comput. Graph.}, 18(1):81--92, 2012.
\newblock \href {https://doi.org/10.1109/TVCG.2010.269}
  {\path{doi:10.1109/TVCG.2010.269}}.

\bibitem{r-05}
D.~Rafiei.
\newblock Effectively visualizing large networks through sampling.
\newblock In {\em VIS 05. IEEE Visualization, 2005.}, pages 375--382, 2005.
\newblock \href {https://doi.org/10.1109/VISUAL.2005.1532819}
  {\path{doi:10.1109/VISUAL.2005.1532819}}.

\bibitem{RS-84}
Neil Robertson and Paul~D. Seymour.
\newblock Graph minors. iii. planar tree-width.
\newblock {\em Journal of Combinatorial Theory, Series B}, 36(1):49--64, 1984.
\newblock URL:
  \url{https://www.sciencedirect.com/science/article/pii/0095895684900133},
  \href {https://doi.org/10.1016/0095-8956(84)90013-3}
  {\path{doi:10.1016/0095-8956(84)90013-3}}.

\bibitem{nr-aaai15}
Ryan~A. Rossi and Nesreen~K. Ahmed.
\newblock The network data repository with interactive graph analytics and
  visualization.
\newblock In {\em Proceedings of the Twenty-Ninth AAAI Conference on Artificial
  Intelligence}, 2015.
\newblock URL: \url{https://networkrepository.com}.

\bibitem{DBLP:journals/tvcg/WuCASQC17}
Yanhong Wu, Nan Cao, Daniel Archambault, Qiaomu Shen, Huamin Qu, and Weiwei
  Cui.
\newblock Evaluation of graph sampling: {A} visualization perspective.
\newblock {\em {IEEE} Trans. Vis. Comput. Graph.}, 23(1):401--410, 2017.
\newblock \href {https://doi.org/10.1109/TVCG.2016.2598867}
  {\path{doi:10.1109/TVCG.2016.2598867}}.

\bibitem{zxyq-13}
Hong Zhou, Panpan Xu, Xiaoru Yuan, and Huamin Qu.
\newblock Edge bundling in information visualization.
\newblock {\em Tsinghua Science and Technology}, 18(2):145--156, 2013.
\newblock \href {https://doi.org/10.1109/TST.2013.6509098}
  {\path{doi:10.1109/TST.2013.6509098}}.

\end{thebibliography}

\clearpage
\appendix

\section{Full Version of the Omitted or Sketched Proofs}\label{app:missingProofs}

\begin{lemma}\label{le:initial-frame-not-maximal}
	There exists a geometric graph $G$ such that in every min-frame optimal story~$\sigma$ of $G$, the initial frame is not a maximal planar subgraph of $G$, even if the intersection graph is a caterpillar.
\end{lemma}
\begin{proof}
    Let $\ell \geq 4$ be an even integer.
    Consider the crossing graph $X$ with $3\ell+3$ vertices indicated in \cref{fig:notmaximal:graph}.  There is a vertex $r$ with $\ell+2$ neighbors $u$, $v$, $w_1,\dots,w_\ell$. The vertices $u$ and $v$ have additional $\ell$ neighbors $u_1,\dots,u_\ell$, and $v_1,\dots,v_\ell$, respectively. Since $X$ is a caterpillar, it is a segment intersection graph and, thus, indeed, a crossing graph of some geometric graph. There is exactly one maximal independent set $I(\{r\})$ containing $r$, see \cref{fig:notmaximal:0}. For each subset $U \subseteq \{u,v\}$,  there is exactly one maximal independent set $I(U)$ of $X$ that contains $U$, but not $r$. 
	Up to symmetry, this yields the three cases indicated with red vertices in \cref{fig:notmaximal:1,fig:notmaximal:2,fig:notmaximal:3}.  Let $F$ be the subgraph of $G$ that corresponds to $I(U)$ for $U=\{r\}$ or $U \subseteq \{u,v\} $ and let $I'(U)$ be the maximum independent set contained in $V(X) \setminus I(U)$. By \cref{cor:bound}, $\mu_{F}(G) \leq \min\{|I(U)|,|I'(U)|\} \leq \ell+2$. On the other hand, if we start with the initial frame $\{u_1,\dots,u_l,v,w_1,\dots,w_{\frac \ell 2}\}$ and add the vertices in the order $v_1,\dots,v_\ell,r,w_{\frac \ell 2+1},\dots,w_\ell,u$, we obtain a planar story $\sigma$ with the frame sizes
	\[
	\frac 32 \ell+1, \frac 32 \ell+1, \dots, \frac 52 \ell, 2\ell + 1, 2\ell+1, \dots,  \frac 52 \ell,  \frac 32 \ell+1,
	\]
	and, thus $\mu(\sigma)=\frac 32 \ell+1$, implying $\mu(G) \geq \frac 32 \ell+1 > \ell+2$. See  \cref{fig:notmaximal:better}.
\end{proof}

\npcomplete*
\label{thm:np-complete*}

\begin{proof}
Problem \textsc{MaxMinFramePlanarStoryD} trivially belongs to NP, since one can non-deterministically find a set of edges $E_1 \subseteq E(G)$ and, in the case $F_1=(V(G), E_1)$ is planar, which can be deterministically tested in polynomial time, one can non-deterministically produce an ordering $e_{1}, e_{2}, \dots e_{|E(G) \setminus F_1|}$ of the edges in $E(G) \setminus E_1$. Then, one can deterministically test in polynomial time whether, for all $i$ such that $1 \leq i \leq \tau$, it holds that $|E(F_i)| \geq m$, where, for $2 \leq j \leq \tau$, $F_j$ is obtained by adding $e_j$ to $F_{j-1}$ and removing the edges of $F_{j-1}$ that intersect $e_j$.
%%%%%%%%%%%%%%%%%%%%%%%%%%%%%%%

To prove the hardness of \textsc{MaxMinFramePla\-narStoryD($G$,$m$)} we apply a reduction from  
\textsc{PlanarMaximumIndependentSetD($H$,$k$)}, i.e., the problem of deciding whether a planar graph $H$ admits an independent set of size at least $k$; this problem is known to be NP-hard~\cite[Theorem 4.1]{m-fcgpa-01}. The reduction is as follows. Define the graph $X$ as the union of $H$ and a disjoint star graph $S$ with $k+1$ leaves (see \cref{fig:npCompleteness}). Since $X$ is planar, it is a segment intersection graph~\cite{chalopinG:stoc09}, i.e., it is a crossing graph of some geometric graph; let $G$ be such a graph. Also, let $m=k+1$. We show that $\mu(G)\geq k+1$ if and only if $H$ contains an independent set of size at least $k$.

Assume first that $H$ contains an independent set $I$ of size at least $k$. Consider the following planar story. The initial frame contains the edges corresponding to $I$ and the center of the star $S$. The size of the initial frame is $k+1$. For the subsequent frames, we add the edges corresponding to the leaves of $S$, in any arbitrary order. This increases the size of the frame up to $2k+1$. Finally, we add in arbitrary order all edges corresponding to vertices of $H$ that were not in $I$. During this process, the edges corresponding to the leaves of $S$ will never be removed. Thus, the size of each frame is at least $k+1$.

Assume now that there is a planar story $\sigma$ of $G$ with $\mu(\sigma) \geq k+1$. We show that $H$ contains an independent set of size at least $k$. Let $F$ be the frame of $\sigma$ that contains the edge corresponding to the center of $S$. In this frame, there is no edge corresponding to a leaf of $S$. This implies that frame $F$ must contain at least $k$ non-intersecting edges that correspond to at least $k$ independent vertices of $H$.
\end{proof}

\begin{theorem}\label{thm:balancedHard}
    Given a graph $X$ and a non-negative integer $m$,
    it is NP-complete to decide whether $G$ contains two disjoint independent sets $I_1$ and $I_2$ such that 
    $\min \{|I_1|,|I_2|\} \geq m$.
\end{theorem}
\begin{proof}
	Containment in NP is obvious. 
	In order to prove hardness, we reduce from NAE 3SAT. A
	\emph{clause} is a set of \emph{literals}, i.e., a set of variables (\emph{positive literals}) and negated variables (\emph{negative literals}). 
	Given a set of clauses, a
	\emph{NAE truth assignment} is a \emph{truth assignment}, i.e., an assignment of true and false to the variables, such that each clause contains at least one true and at least one false literal. The problem NAE 3SAT takes as input a 3SAT formula, i.e., a set of clauses each containing three literals, and asks whether there is a NAE truth assignment. Observe that NAE 3SAT is NP-complete in general \cite{darmann/doecker:20},
    but polynomial-time solvable if the variable-clause graph is planar \cite{moret:88}.

\begin{figure}[htbp]
    \centering
    \includegraphics[page=5]{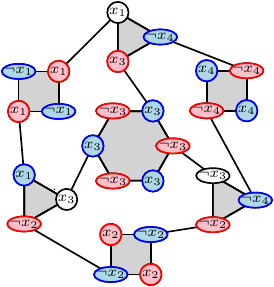}
    \caption{
        Graph $G_I$ for the 3-SAT formula 
        with $p=3$ clauses $\{x_1,x_3,\neg x_4\}$, $\{x_1, x_3,\neg x_2\}$, and $\{\neg x_2, \neg x_3,\neg x_4\}$. The red vertices and the blue vertices, respectively, form an independent set of size $4p$ corresponding to the NAE truth assignment $x_1,x_2 \rightarrow \textsc{false}$, 
        %$x_2 \rightarrow \textsc{false}$, 
        $x_3,x_4 \rightarrow \textsc{true}$.
        %, and $x_4 \rightarrow \textsc{true}$.
    }
    \label{fig:balancedHard}
\end{figure}

	Consider now an instance $I=(\mathcal X,\mathcal C)$ of 3SAT with $p$ clauses. We construct a graph $G_I$ such that $I$ is a yes instance of NAE 3SAT if and only if $G_I$ contains two disjoint independent sets, each of size at least $4p$. See \cref{fig:balancedHard}. For each clause $c$, graph $G_I$\label{def:gi} contains a \emph{clause cycle} of three vertices, each labeled with one of the literals of $c$. 
    For each variable $x$ there is a \emph{variable cycle} $C_x:v_1,\dots,v_{2d}$ of length $2d$, where $d$ is the number of clauses that contain a literal $x$ or $\neg x$. For $i=1,\dots,d$, we label $v_{2i-1}$ with $x$ and $v_{2i}$ with~$\neg x$. Let $c_j$, $j=1,\dots,k$ be the clauses containing a literal $x$ or~$\neg x$ and let $w_j$ be the vertex of $C_{c_j}$ labeled with $x$ or~$\neg x$. We connect $w_j$ to $v_{2j-1}$ or $v_{2j}$ such that $w_j$ and its neighbor in $C_x$ have the same~label.
	
	Assume first that there is an NAE truth assignment for $I=(\mathcal X,\mathcal C)$. Let $I_1$ be the independent set that contains all negative literals from the variable cycles and one positive literal from each clause cycle. Analogously, let $I_2$ be the independent set that contains all positive literals from the variable cycles and one negative literal from each clause cycle. Both, $I_1$ and $I_2$ have exactly $4p$ vertices. 
	
    Assume now that $G_I$ contains two disjoint independent sets $I_1$ and $I_2$, each of size at least $4p$. Observe that an independent set contains at most every second vertex out of any cycle. In particular, it contains at most $3p$ vertices out of the $6p$ vertices in all variable cycles and at most one vertex out of each clause triangle. Thus, in order to reach $4p$ vertices, both $I_1$ and $I_2$ contain exactly one vertex out of each clause triangle and exactly every second vertex out of each variable cycle. More precisely, for each variable $x$, the set $I_1$ either contains all positive or all negative literals of $x$ in the variable cycle $C_x$ of $x$, and $I_2$ contains the other literals in $C_x$. We set a variable to true if the positive literals are in $I_2$ and false otherwise; i.e., all vertices of a variable cycle that are in $I_2$ (resp. $I_1$) correspond to a true (resp. false) literal. This yields an NAE truth assignment: Each clause cycle contains a vertex $v_1 \in I_1$ and a vertex $v_2 \in I_2$. Since the neighbor of $v_1$ in the variable cycle cannot be in $I_1$, we have that $v_1$ corresponds to a true literal. Analogously, $v_2$ corresponds to a false literal.
\end{proof}

\maximumPairFPT* \label{thm:maximumPairFPT*}
\begin{proof}
    The proof follows the ideas of the FPT approach for constructing a minimum vertex cover proposed by Niedermeier~\cite{niedermeier_book},  enhancing it with the concept of Pareto optimal pairs.
    Let~$\mathcal T$ be a tree decomposition of width~$\omega$ of a graph~$X$  with $n$ vertices.
    We root $\mathcal T$ at an arbitrary vertex. We use dynamic programming. 
    For a vertex~$\nu$ of~$\mathcal T$, let $\mathcal T(\nu)$ be the subtree of $\mathcal T$ rooted at~$\nu$, let $X(\nu)$ be the subgraph of~$X$ induced by the vertices in the bags of~$\mathcal T(\nu)$, and let~$n(\nu)$ be the number of vertices of~$X(\nu)$. For each vertex~$\nu$ of~$\mathcal T$ and for each map $c:V(\nu)\rightarrow \{$red, blue, white$\}$, let $I(\nu,c)$ be the set of pairs~$(\alpha,\beta)$ for which $X(\nu)$ contains two disjoint independent sets $I_1$ and $I_2$ of size~$\alpha$ and~$\beta$, respectively, such that the red vertices of~$V(\nu)$ are contained in~$I_1$, the blue vertices of~$V(\nu)$ in~$I_2$, and the white vertices of $V(\nu)$ are neither in $I_1$ nor~$I_2$. Since a subset of an independent set is an independent set, it suffices to store the \emph{Pareto optimal} pairs, i.e., the pairs 
    $(\alpha,\beta) \in I(\nu,c)$ for which there is no pair $(\alpha',\beta') \in I(\nu,c)$ with $\alpha \leq \alpha'$ and $\beta < \beta'$ or $\alpha < \alpha'$ and $\beta \leq \beta'$. Let $L(\nu,c)$ be the list of Pareto-optimal pairs from $I(\nu,c)$. Observe that $|L(\nu,c)| \in \mathcal O(n(\nu))$ and that the number of functions $c:V(\nu)\rightarrow \{$red, blue, white$\}$ is in~$\mathcal O(3^\omega)$.

    We traverse $\mathcal T$ bottom-up.
    Let $\nu$ be a vertex of $\mathcal T$ and let $c:V(\nu)\rightarrow \{$red, blue, white$\}$. If the red and blue vertices of $V(\nu)$, respectively, form an independent set of $X$, we initialize $L(\nu,c)$ with $(|\{v \in V(\nu);\; c(v) = \textup{red}\}|,|\{v \in V(\nu);\; c(v) = \textup{blue}\}|)$. Otherwise, $L(\nu,c)=\emptyset$. The independence test can be done in $\mathcal O(\omega^2)$ time.
    
    Let $\mu_1,\dots,\mu_\ell$ be the children of $\nu$, if any. 
    We iteratively add $\mathcal T(\mu_i)$, $i=1,\dots,\ell$ by combining 
    the current list $L(\nu,c)$ and the lists $L(\mu_i,c_i)$ for all colorings $c':V(\mu_i)\rightarrow \{$red, blue, white$\}$ that are \emph{compatible} with $c$, i.e., that fulfill $c'(v)=c(v)$ for $v \in V(\nu) \cap V(\mu_i)$: Let $\alpha_i = |\{v \in V(\nu) \cap V(\mu_i);\; c(v) = \textup{ red}\}|$ and $\beta_i = |\{v \in V(\nu) \cap V(\mu_i);\; c(v) = \textup{ blue}\}|$.
     For each  $(\alpha,\beta) \in L(\nu,c)$, each coloring $c':V(\mu_i)\rightarrow \{$red, blue, white$\}$ that is compatible with $c$, and each $(\alpha',\beta') \in L(\mu_i,c')$, we add $(\alpha+\alpha'-\alpha_i,\beta+\beta'-\beta_i)$ to some initially empty list~$L_i$. After having computed these potentially $\mathcal O(n^2)$ pairs for $\nu$, $\mu_i$ and $c$, we replace the current list $L(\nu,c)$ with the  Pareto-optimal entries of $L_i$.

    In order to have direct access to the lists $L(\mu,c')$ corresponding to the colorings $c'$ that are compatible with a coloring $c$ of the parent $\nu$ of $\mu$, we store the lists $L(\mu,\cdot)$ in an array~$A(\mu)$. To this end, consider an ordering $v_0,\dots,v_\omega$ of the elements in~$V(\mu)$ such that the elements of $V(\nu) \cap V(\mu)$ are last. The entry in~$A(\mu)$ with position $\sum_{i=0}^\omega \gamma_i 3^i$ then contains the list~$L(\mu,c')$ where $c'$ colors~$v_i$ white if $\gamma_i = 0$, red if $\gamma_i=1$, and blue if~$\gamma_i=2$.
    Now, each iteration can be computed in $\mathcal O(n^2)$ time and we need $\deg \nu$ steps in order to compute the final list $L(\nu,c)$. By the hand-shaking lemma, the run time for computing the lists $L(\nu,c)$ for all vertices $\nu$ of $\mathcal T$ and for all colorings~$c$ of the bags is in $\mathcal O(3^\omega \cdot \sum_{\nu \in V(\mathcal T)} (\omega^2+\omega \cdot n^2 \cdot \deg \nu)) \subseteq \mathcal O(3^\omega\omega^2n^3)$.
        
    In the end,  it suffices to go through the list of Pareto-optimal elements of the root and to extract the entry where $\min \{\alpha,\beta\}$ is maximum. Tracing back the way this solution was obtained, we can also derive the respective pair of independent sets.
\end{proof}

\subsection{3-Plane and 2-Plane Geometric Graphs}\label{app:3-planar}

We use the following lemma in order to prove that \textsc{MaxMinFramePlanarStory($G$)} can be solved in polynomial time if $G$ is 3-plane and the crossing graph is cycle-free. 

\heurOptCases* \label{thm:heurOptCases*}

\begin{proof}
	We formulate the proof in terms of the crossing graph $X$. 
	 We say that a vertex is red if it is contained in $I_1$, blue if it is contained in $I_{\tau}$,  and white if it is neither in $I_1$ nor in $I_\tau$.  Let $\sigma = \langle I_1,\dots,I_\tau\rangle$ be the story computed by the  advanced greedy heuristic. For $i=2,\dots,\tau$, let $v_i$ be the only vertex in $I_i \setminus I_{i-1}$. We say that a vertex is $i$-current or $i$-future, respectively, if it was current or future  at the time when we choose $v_i$ for $i=2,\dots,\tau$. Observe that the vertices in $I_{i-1}$ are the $i$-current vertices and the vertices in $(I_i \cup \dots \cup I_\tau) \setminus I_{i-1}$ are the $i$-future vertices.  Let $X_i$ be the subgraph of $X$ induced by the vertices in $I_{i-1} \cup \dots \cup I_\tau$. The $i$-current degree of a vertex $v$ of $X_i$  is the number of neighbors of $v$ in $I_{i-1}$. 
	
	Let $t$ be the lowest index among $2,\dots,\tau$ such that the $t$-current degree of $v_t$ is greater than one. We have $|I_i| \geq |I_1|$ for $i=1,\dots,t-1$.  	
	\begin{claim}
		There is no index $j$ with $t \leq j \leq \tau$ such that the $j$-current degree of $v_j$ is zero. 
		\end{claim}
	The claim implies that from $I_{t-1}$ on the frame sizes can only decrease or stay the same until we reach the final frame $I_\tau$. This implies that $|I_j| \geq |I_\tau|$ for $j=t-1,\dots,\tau$. 	
		
	It remains to prove the claim.  To this end
	let $H$ be a connected component of $X_j$. 
	We will show by induction on $j=t,\dots,\tau$ that 
	\begin{inparaenum}[(a)]
		\item
		each admissible $j$-future vertex of $H$ has $j$-current degree zero (which proves the claim) and
		\item 
		at most one admissible $j$-future vertex of $H$ has $j$-current degree one.
	\end{inparaenum}
		
	If $j=t$ then every admissible $j$-future vertex has at least two $j$-current neighbors.
	Let now $j>t$. Note that $X_j$ is obtained from $X_{j-1}$ by removing the $(j-1)$-current neighbors of $v_{j-1}$. Let $H'$ be the connected component of $X_{j-1}$ that contains $H$. 
	By the inductive hypothesis, no admissible $(j-1)$-future vertex of $H$ has $(j-1)$-current degree zero and
	at most one admissible $(j-1)$-future vertex of $H$ has $(j-1)$-current degree one.
	If $H = H'$ we are done, as in this case the $j$-current vertices and the $(j-1)$-current vertices coincide on $H'$.
	
	So assume that $v_{j-1}$ is a vertex of $H'$. Then $H$ is a connected component of the graph  obtained from $H'$ by removing the $(j-1)$-current neighbors of $v_{j-1}$. If $H$ does not contain any $j$-future vertices, there is nothing to show. So assume that $H$ contains some $j$-future vertex. In particular $H$ does not only contain $v_{j-1}$.
	Consider first the case that $v_{j-1}$ is a blue vertex. Then all neighbors of $v_{j-1}$ are $(j-1)$-current and $H$ contains exactly one vertex~$v$ that was a neighbor of some neighbor $w$ of $v_{j-1}$. No vertex of $H$ became newly admissible and the $j$-current degree equals the  $(j-1)$-current degree of  any vertex of~$H$, except for $v$.
	
	Observe that no admissible vertex of $H$ has $(j-1)$-current degree one: If $v_{j-1}$ had $(j-1)$-current degree one in $X_{j-1}$, then, by the inductive hypothesis, it was the only admissible vertex  of $H'$ that had $(j-1)$-current degree one. If $v_{j-1}$ had $(j-1)$-current degree greater than one, then, by the greedy strategy, all admissible $(j-1)$-future vertices had $(j-1)$-current degree greater than one. Since the $j$-current degree of $v$ equals its $(j-1)$-current degree minus one, 
	this implies that if $v$ was admissible then the $j$-current degree of $v$ is not zero and that $v$ is the only admissible $j$-future vertex of  $j$-current degree 1 of $H$, if any.
		
	If $v_{j-1}$ was a white vertex then $v_{j-1}$ has exactly one blue neighbor $v^b$ and all other neighbors of $v_{j-1}$ are $(j-1)$-current: Since every white vertex has at least one blue and at least one red neighbor, it follows  that the white vertices that are adjacent to another white vertex initially have current degree one.  Thus, among any two adjacent white vertices, at least one was chosen before $v_t$. Moreover, any white vertex in $X_j$ with two blue neighbors was also chosen before $v_t$.

    If $H$ is the component containing $v_{j-1}$ then $H$ contains at least two vertices, namely $v_{j-1}$ and $v^b$. Observe that $v_{j-1}$ is not a $j$-future vertex. However, $v^b$ might have become an admissible $j$-future vertex. If $v^b$ is admissible and has $j$-current degree one then $v^b$ and $v_{j-1}$ are the only vertices of $H$, and, thus, $v^b$ is the only vertex of $H$ with $j$-current degree one.
	
	If $H$ is not the component containing $v_{j-1}$, then there is again exactly one vertex $v$ that was a neighbor of a $(j-1)$-current neighbor $w$ of $v_{j-1}$.
	By the same argument as above,  no admissible vertex of $H$ has $(j-1)$-current degree one, and, thus, $v$ has $j$-current degree greater one and $v$ is the only admissible $j$-future vertex of  $j$-current degree 1 of $H$, if any.
\end{proof}

\TwoandThreeplanar* \label{thm:2and3planar*}
\begin{proof}
    For a 3-plane geometric  graph the statement follows immediately from \cref{thm:heurOptCases}.  
    So, let $X$ be the crossing graph of a 2-plane geometric graph. The degree of each vertex in $X$ is at most two. Thus, each connected component of $X$ is either a path or a simple cycle. 

    A path or a cycle of $X$ is \emph{odd} if it contains an odd number of vertices, otherwise it is \emph{even}. A maximum pair $I_1,I_\tau$ of independent sets looks as indicated in \cref{fig:2planar}.
    Divide the set of odd paths into two sets $\mathcal O_1$ and $\mathcal O_2$ such that $0 \leq |\mathcal O_2| - |\mathcal O_1| \leq 1$.
    See \cref{fig:2planar_oddPath2} for the paths in $\mathcal O_1$ and \cref{fig:2planar_oddPath1} for the paths in $\mathcal O_2$.
    %Choose an initial frame $X[F_1]$ such that: 
    The initial frame $I_1$ contains
    $(i)$ a maximum independent set for each cycle and for each path not in $\mathcal O_2$; $(ii)$ for each path in $\mathcal O_2$ all the vertices that are not in the maximum independent set (i.e., every second vertex of the path, omitting the first and the last vertex). See the red vertices in \cref{fig:2planar}. The vertices in $I_\tau$ (blue vertices in \cref{fig:2planar}) are all vertices that are not red, except for one vertex in each odd~cycle.

    \textsf{Advanced Greedy 1a}  first adds all vertices of current degree one or zero. These are all non-red vertices in the odd cycles, the even paths, and the odd paths in $ \mathcal O_2$. No step will decrease the size of a  frame, and the addition of the final vertex of a path in $ \mathcal O_2$ increases the size of the frame by one. Finally,  \textsf{Advanced Greedy 1a} adds the vertices of the even cycles and the odd paths in $\mathcal O_1$. Moreover, after a first vertex of an even cycle or an odd path in $\mathcal O_1$ is picked, all other vertices of that component have current degree at most one and are chosen immediately afterwards. The first vertex of each such cycle or path decreases the number of edges in a frame by one. The last vertex of an even cycle increases it again by one. Thus, a cycle only temporarily drops the frame size by one. In the end the frame size is $|I_\tau|$.
    
    Summarizing, if there is at least one odd path or if there is no even-length cycle then the minimum size of a frame is $\min \{|I_1|,|I_\tau|\}$; otherwise it is $\min\{|I_1|,|I_\tau|\}-1$. In either case, this is optimum. Indeed, by \cref{cor:pair}, the minimum frame size cannot be greater than $\min\{|I_1|,|I_\tau|\}$.
    If there is no odd path, then no independent set of $X$ contains more than $|I_1|=|I_\tau|$ vertices. 
    Also, if there is an even cycle $C$ then the frame size drops by one when adding the first vertex of $C$ not in $I_1$.
\end{proof}

%%%%%%%%%%%%%%%%%%% APP EXPERIMENTS

\clearpage
\section{Additional Material for \Cref{se:experiments}}\label{se:app-experiments}

%%%% TABLE RUNTIME ON REAL-GRAPHS
\begin{table}[h]
\centering
\begin{tabular}{|>{\raggedright\arraybackslash}p{3cm}|r|r|r|r|r|r|}
\hline
\textbf{graph name} & \textsf{AG-1a2a} & \textsf{AG-1a2b} & \textsf{AG-1b2a} & \textsf{AG-1b2b} & \textsf{AG-1c2a} & \textsf{AG-1c2b} \\
\hline\hline
bwm200 & 0.00 & 0.00 & 0.00 & 0.00 & 0.00 & 0.00 \\
\hline
ca-sandi\_auths & 0.00 & 0.00 & 0.00 & 0.00 & 0.00 & 0.00 \\
\hline
rajat11 & 0.09 & 0.13 & 0.06 & 0.08 & 0.10 & 0.13 \\
\hline
lesmis &  &  & 0.05 & 0.05 & 0.08 & 0.08 \\
\hline
ca-netscience &  &  & 0.61 & 1.00 & 1.00 & 1.00 \\
\hline
gd06\_theory &  &  & 0.07 & 0.07 & 0.07 & 0.07 \\
\hline
road-chesapeake &  &  & 0.04 & 0.04 & 0.05 & 0.05 \\
\hline
insecta-beetle-group-c1-period-1 &  &  & 0.04 & 0.04 & 0.06 & 0.06 \\
\hline
polbooks &  &  & 0.32 & 0.32 & 0.53 & 0.53 \\
\hline
email-enron-only &  &  & 1.05 & 1.05 & 1.64 & 1.64 \\
\hline
eco-stmarks &  &  & 0.32 & 0.32 & 0.40 & 0.40 \\
\hline
adjnoun &  &  & 0.53 & 0.53 & 0.70 & 0.70 \\
\hline
\end{tabular}
\caption{Runtime in seconds for the different heuristics on the {\tt real graphs}. The empty cells refer to the instances that could not be solved by Variant~\textsf{1a} with the given time limit.}
\label{tb:running-time-real-graphs}
\end{table}

\clearpage

%%%% FIGURE RUNTIME
\begin{figure}[h!]
	%\begin{subfigure}[b]{.8\linewidth}
		\centering 
       \includegraphics[page=1,width=0.8\columnwidth]{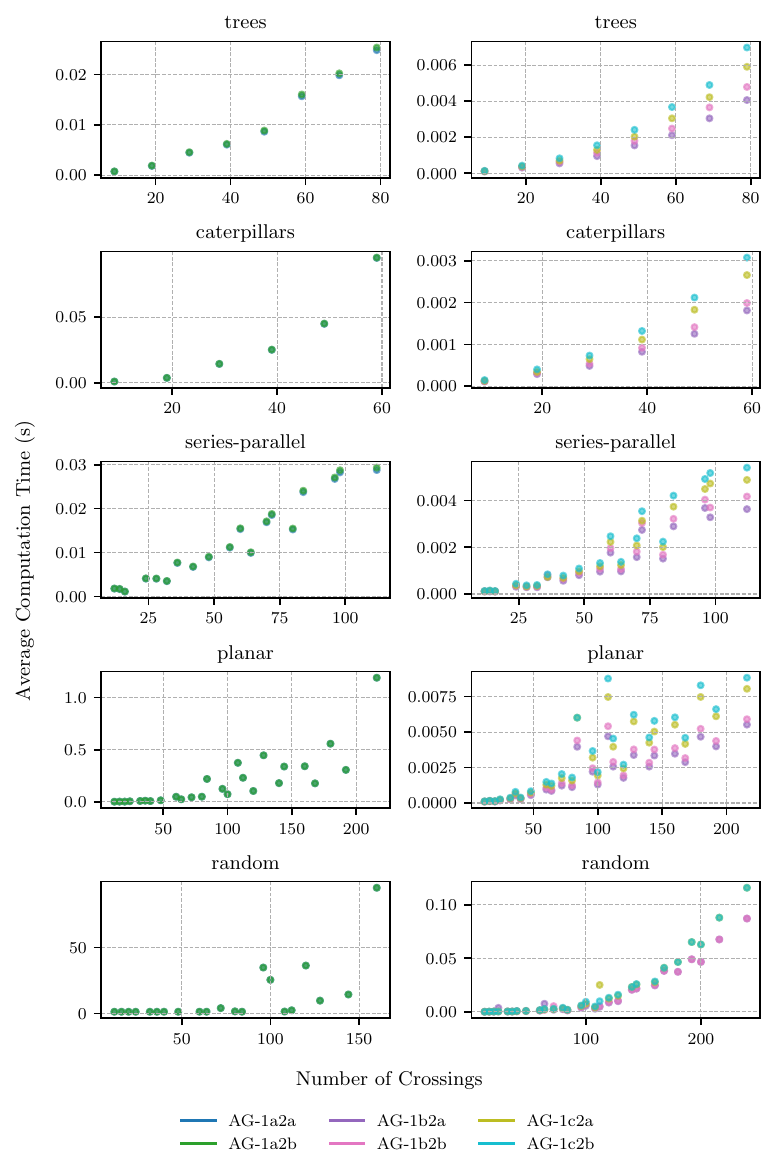} 
	%\end{subfigure}
	\caption{\label{fig:time_plots} 
    %\textcolor{red}{I suggest to split each chart into two charts: one comparing (1b2a, 1b2b) against (1c2a, 1c2b); the other comparing 1a2a and 1a2b.} 
    Runtime (in seconds) of the heuristics by increasing number of crossings (edges of the crossing graph). The charts on the left correspond to \textsf{AG-1a2a} and \textsf{AG-1a2b}; the trend lines largely overlap, as the differences between Variants \textsf{2a} and \textsf{2b} are negligible. The charts on the right refer to the remaining heuristics. Due to the scale, a slight difference between Variants \textsf{2a} and \textsf{2b} can be observed—except in the case of {\tt random graphs}.}
    %\todo[inline]{SC: For Variants 1a: Does a linear interpolation make sense if we expect a cubic run time (even provided the width of the tree decomposition is constant)?}}
\end{figure}

\clearpage

%%%% TABLE DETAILS HEURISTICS
\begin{table}[tb]
  	\centering
  \begin{tabular}{|l || p{0.07\linewidth}| p{0.08\linewidth} | p{0.09\linewidth}| p{0.08\linewidth}|| p{0.07\linewidth}| p{0.08\linewidth} | p{0.09\linewidth}| p{0.08\linewidth}|}
  	\hline
  	& \multicolumn{4}{r||}{\textbf{Caterpillars}}& \multicolumn{4}{r|}{\textbf{General Trees}}\\
  	\hline
  	& \textsf{\% Opt} & \textsf{Avg approx. ratio} & \textsf{Min-Max approx. ratio} & \textsf{SD approx. ratio} & \textsf{\% Opt} & \textsf{Avg approx. ratio} & \textsf{Min-Max approx. ratio} & \textsf{SD approx. ratio} \\
  	\hline
  	\textsf{AG-1a2a} & $97.0\%$ (32/33) & $99.64\%$ & $88.24\%$ \ $100.00\%$ & $2.02\%$ & $100.0\%$ (54/54) & $100.00\%$ & $100.00\%$ \ $100.00\%$ & $0.00\%$ \\
  	\hline
  	\textsf{AG-1a2b} & $97.0\%$ (32/33) & $99.64\%$ & $88.24\%$ \ $100.00\%$ & $2.02\%$ & $100.0\%$ (54/54) & $100.00\%$ & $100.00\%$ \ $100.00\%$ & $0.00\%$ \\
  	\hline
  	\textsf{AG-1b2a} & $3.0\%$ (1/33) & $58.67\%$ & $42.86\%$ \ $100.00\%$ & $12.28\%$ & $0.0\%$ (0/54) & $65.80\%$ & $47.37\%$ \ $88.89\%$ & $8.35\%$ \\
  	\hline
  	\textsf{AG-1b2b} & $3.0\%$ (1/33) & $58.67\%$ & $42.86\%$ \ $100.00\%$ & $12.28\%$ & $0.0\%$ (0/54) & $65.87\%$ & $47.37\%$ \ $88.89\%$ & $8.26\%$ \\
  	\hline
  	\textsf{AG-1c2a} & $18.2\%$ (6/33) & $81.32\%$ & $64.29\%$ \ $100.00\%$ & $10.10\%$ & $14.8\%$ (8/54) & $84.43\%$ & $73.33\%$ \ $100.00\%$ & $7.87\%$ \\
  	\hline
  	\textsf{AG-1c2b} & $18.2\%$ (6/33) & $81.66\%$ & $64.29\%$ \ $100.00\%$ & $10.17\%$ & $14.8\%$ (8/54) & $84.43\%$ & $73.33\%$ \ $100.00\%$ & $7.87\%$ \\
  	\hline
  \end{tabular}
  \centering
  \begin{tabular}{|l || p{0.07\linewidth}| p{0.08\linewidth} | p{0.09\linewidth}| p{0.08\linewidth} || p{0.07\linewidth}| p{0.08\linewidth} | p{0.09\linewidth}| p{0.08\linewidth}|}
  	\hline
  	& \multicolumn{4}{r||}{\textbf{Series-Parallel Graphs}}	& \multicolumn{4}{r|}{\textbf{Planar Graphs}}\\
  	\hline
  	& \textsf{\% Opt}  & \textsf{Avg approx. ratio} & \textsf{Min-Max approx. ratio} & \textsf{SD approx. ratio} & \textsf{\% Opt} & \textsf{Avg approx. ratio} & \textsf{Min-Max approx. ratio} & \textsf{SD approx. ratio} \\
  	\hline
  	\textsf{AG-1a2a} & $26.6\%$ (25/94) & $90.66\%$ & $66.67\%$ \ $100.00\%$ & $9.29\%$  & $40.7\%$ (46/113) & $88.91\%$ & $50.00\%$ \ $100.00\%$ & $13.48\%$ \\
  	\hline
  	\textsf{AG-1a2b} & $33.0\%$ (31/94) & $91.25\%$ & $66.67\%$ \ $100.00\%$ & $9.79\%$& $41.6\%$ (47/113) & $90.06\%$ & $50.00\%$ \ $100.00\%$ & $12.38\%$ \\
  	\hline
  	\textsf{AG-1b2a} & $3.2\%$ (3/94) & $76.17\%$ & $50.00\%$ \ $100.00\%$ & $9.72\%$ & $6.2\%$ (7/113) & $61.36\%$ & $33.33\%$ \ $100.00\%$ & $13.24\%$ \\
  	\hline
  	\textsf{AG-1b2b} & $2.1\%$ (2/94) & $76.30\%$ & $50.00\%$ \ $100.00\%$ & $9.07\%$ & $6.2\%$ (7/113) & $61.93\%$ & $42.86\%$ \ $100.00\%$ & $12.62\%$ \\
  	\hline
  	\textsf{AG-1c2a} & $5.3\%$ (5/94) & $81.41\%$ & $62.50\%$ \ $100.00\%$ & $7.90\%$ & $10.6\%$ (12/113) & $78.99\%$ & $44.44\%$ \ $100.00\%$ & $12.47\%$ \\
  	\hline
  	\textsf{AG-1c2b} & $7.4\%$ (7/94) & $81.58\%$ & $50.00\%$ \ $100.00\%$ & $9.26\%$ & $12.4\%$ (14/113) & $80.73\%$ & $50.00\%$ \ $100.00\%$ & $12.06\%$ \\
  	\hline
  \end{tabular}
  \centering
  \begin{tabular}{|l || p{0.07\linewidth}| p{0.08\linewidth} | p{0.09\linewidth}| p{0.08\linewidth}|| p{0.07\linewidth}| p{0.08\linewidth} | p{0.09\linewidth}| p{0.08\linewidth}|}
  	\hline
  	& \multicolumn{4}{r||}{\textbf{Random Graphs}}& \multicolumn{4}{r|}{\textbf{Real Graphs}}\\
  	\hline
  	& \textsf{\% Opt} & \textsf{Avg approx. ratio} & \textsf{Min-Max approx. ratio} & \textsf{SD approx. ratio}	& \textsf{\% Opt} & \textsf{Avg approx. ratio} & \textsf{Min-Max approx. ratio} & \textsf{SD approx. ratio} \\
  	\hline
  	\textsf{AG-1a2a} & $76.4\%$ (84/110) & $97.83\%$ & $82.61\%$ \ $100.00\%$ & $4.43\%$ & $100.0\%$ (2/2) & $100.00\%$ & $100.00\%$ \ $100.00\%$ & $0.00\%$ \\
  	\hline
  	\textsf{AG-1a2b} & $81.8\%$ (90/110) & $98.30\%$ & $83.33\%$ \ $100.00\%$ & $4.09\%$  & $100.0\%$ (2/2) & $100.00\%$ & $100.00\%$ \ $100.00\%$ & $0.00\%$ \\
  	\hline
  	\textsf{AG-1b2a} & $10.9\%$ (12/110) & $71.32\%$ & $46.67\%$ \ $100.00\%$ & $13.47\%$& $0.0\%$ (0/2) & $71.67\%$ & $60.00\%$ \ $83.33\%$ & $11.67\%$ \\
  	\hline
  	\textsf{AG-1b2b} & $10.9\%$ (12/110) & $71.68\%$ & $46.67\%$ \ $100.00\%$ & $13.30\%$& $100.0\%$ (2/2) & $100.00\%$ & $100.00\%$ \ $100.00\%$ & $0.00\%$ \\
  	\hline
  	\textsf{AG-1c2a} & $22.7\%$ (25/110) & $87.10\%$ & $66.67\%$ \ $100.00\%$ & $8.91\%$ & $100.0\%$ (2/2) & $100.00\%$ & $100.00\%$ \ $100.00\%$ & $0.00\%$ \\
  	\hline
  	\textsf{AG-1c2b} & $22.7\%$ (25/110) & $87.59\%$ & $64.71\%$ \ $100.00\%$ & $8.80\%$ & $100.0\%$ (2/2) & $100.00\%$ & $100.00\%$ \ $100.00\%$ & $0.00\%$ \\
  	\hline
  \end{tabular}
  \caption{The performance of our heuristics in comparison to the exact algorithm on each type of graphs.
  }
  \label{tab:heuristic_performance_separate}
  \end{table}

\clearpage

\begin{figure}[tb]
    \centering
    \begin{subfigure}{.8\columnwidth}
        \centering
        \includegraphics[page=1,width=\columnwidth]{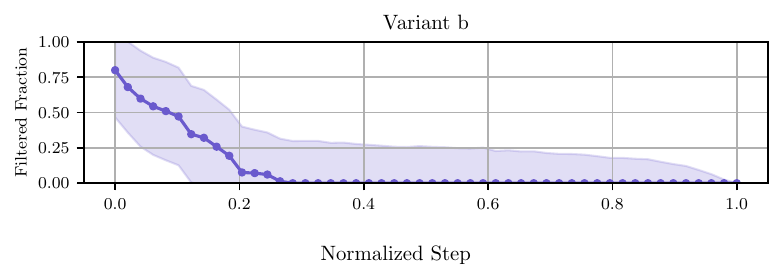}
        \subcaption{}
        \label{fig:candidate_set}
    \end{subfigure}
    \hfil
    \begin{subfigure}{.8\columnwidth}
        \centering
        \includegraphics[page=1,width=\columnwidth]{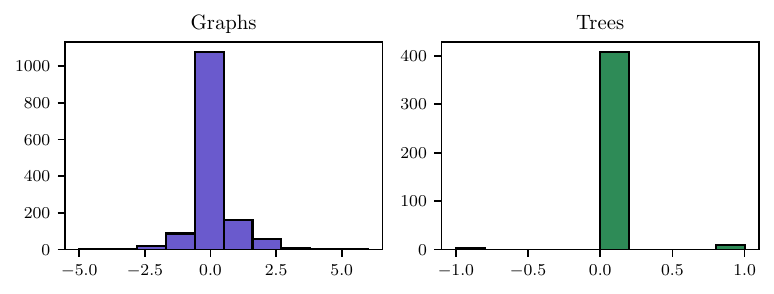}
        \subcaption{}
        \label{fig:histograms_differences}
    \end{subfigure}
    \hfil
    \caption{(a) The median of the fractions of nodes filtered by Variant~\textsf{2b} per normalized algorithm iteration. Blue curve: Median of the percentages of nodes filtered at each step. Shaded region: Standard deviation, indicating variability across graphs. (b) Histograms for the values of~$D_G$ and~$D_T$.}\label{fi:variants2a2b}
\end{figure}

\subparagraph{Differences between Variants \textsf{2a} and \textsf{2b}.} To better understand the behavior of the different variants for executing \textsf{Phase~2}, 
we investigated two aspects: (1) whether tie-breaking actually occurs—i.e., whether the set $E'$ of admissible future edges with minimum current degree often contains multiple candidates; and (2) if so, whether applying the additional tie-breaking rule of Variant~\textsf{2b} is more effective than selecting an edge from $E'$ uniformly at random. As shown in \cref{fig:candidate_set}, especially during the initial iterations of the greedy algorithm, Variant~\textsf{2b} consistently reduces the size of $E'$—i.e., the refined subset $E'' \subseteq E'$ (as defined in \cref{se:heuristic}) is noticeably smaller. Having established that tie-breaking situations do arise, we then examined whether Variant~\textsf{2b} empirically produces solutions equivalent to those of Variant~\textsf{2a}.

For each graph $G_i$, let $\mu^{b}(G_i)$ and $\mu^{a}(G_i)$ be the values of the solution computed by Variant~\textsf{2b} and Variant~\textsf{2a}, respectively. Also, denote by $D(G_i)=\mu^{b}(G_i)-\mu^{a}(G_i)$. 
%We then formed the paired differences
% \todo{Shall we say for $y \in \{a,b,c\}$ let \dots EK: We can, depends on whether we will drop the three different variants}
%$D(G_i)=\mu^{b}(G_i)-\mu^{a}(G_i)$.
Because our greedy heuristic showed a better performance for trees, we split the analysis into two groups: the set $T$ of graphs $G_i$ that are trees (i.e. caterpillar graphs and general trees), and the set $G$ of graphs $G_i$ that are not trees.
Also, let $D_G=\{D(G_i) \; | \; G_i \in T\}$ and $D_G=\{D(G_i) \; | \; G_i \in G\}$. Because a large amount of our observed $D(G_i)$ values were exactly zero (creating a heavy spike at zero, see \cref{tab:zeros} and Figure \ref{fig:histograms_differences}), the distributions of the differences are non‐normal.  We therefore used a non-parametric bootstrap to estimate a $95\%$ confidence interval (CI) for the medians of $D_T$ and $D_G$.  Specifically:
\begin{enumerate}
	\item For each difference‐set (i.e. $D_T$ and $D_G$), create $B=10,000$ bootstrap samples by sampling with replacement from the original data.
	\item Compute the median of each bootstrap sample.
	\item Take the $2.5$th and $97.5$th percentiles of these $10,000$ medians to form a $95\%$ CI.
\end{enumerate}

In both cases, the $95 \%$ bootstrap CI for the median collapsed to $[0.00, 0.00]$. We 
conclude that Variant~\textsf{2b} does not have a statistically or practically significant advantage over breaking ties at random~--~regardless of whether the input is a tree or not.

\begin{table}[h!]
	\centering
	\begin{tabular}{|l ||r| |r|}
		\hline
	     &	Non-trees, $D_G$ & Trees, $D_T$\\
		                 	\hline
		Sample size                    & 1449 &  420 \\
		\hline
		Count of zero differences      & 1103 &408   \\
		\hline
		Proportion of zeros            & 0.76   & 0.97  \\
		\hline
	\end{tabular}
	\caption{Empirical equivalence test: The median difference between Variants~\textsf{2b} and~\textsf{2a}.}
	\label{tab:zeros}
\end{table}

\end{document}